\documentclass[aps,prx,twocolumn,superscriptaddress,longbibliography,nofootinbib,floatfix]{revtex4-2}

\usepackage[T1]{fontenc}
\usepackage[utf8]{inputenc}
\usepackage{lmodern}
\usepackage{amssymb}

\usepackage{graphicx}
\usepackage{dcolumn}
\usepackage{bm}
\usepackage{xcolor}
\usepackage{amsthm}
\usepackage{physics}
\usepackage{mathtools}
\usepackage[disable=footnote]{microtype}
\usepackage{placeins}
\usepackage{tikz}
\usetikzlibrary{arrows.meta}
\allowdisplaybreaks

\theoremstyle{plain}
\newtheorem{lemma}{Lemma}
\newtheorem{proposition}{Proposition}
\newtheorem{theorem}{Theorem}
\newtheorem{corollary}{Corollary}
\newtheorem{result}{Result}
\begin{document}

\title{Semi-Device-Independent Quantum Key Distribution from Operational Assumptions}
\author{Anubhav Chaturvedi}
\email{anubhav.chaturvedi@ug.edu.pl}
\affiliation{Division of Quantum Optics and Information, Institute of Theoretical Physics and Astrophysics, Faculty of Mathematics, Physics and Informatics, University of Gda\'nsk, 80-308 Gda\'nsk, Poland}
\author{Giuseppe Viola}
\affiliation{Naturwissenschaftliche-Technische Fakultät, Universität Siegen, 57068 Siegen, Germany}
\author{Ekta Panwar}
\affiliation{Slovak Academy of Sciences, Institute of Physics, 84104, Bratislava, Slovakia}
\author{Tushita Prasad}
\affiliation{International Centre for Theory of Quantum Technologies, University of Gda\'nsk, 80-309 Gda\'nsk, Poland}
\author{Debashis Saha}
\affiliation{Department of Physics, School of Basic Sciences, Indian Institute of Technology Bhubaneswar, Odisha 752051, India}

\begin{abstract}

Semi-device-independent quantum key distribution leaves the measurement devices uncharacterized while placing a trusted assumption on Alice's source. We formulate this source assumption operationally on Alice's four-preparation ensemble as a scalar bound on one of four physically motivated operational source tasks: four-state discrimination, parity discrimination, or their normalized composites with state exclusion. For the two-bit random-access code, we derive the exact classical frontier for each of the four source assumptions. Numerically, the BB84 strategy attains the maximal quantum deviation from all four frontiers, while the preparation-depolarized BB84 family and the direct-sum leakage family trace complementary branches of the arbitrary-dimensional quantum boundary for the two exclusion-assisted assumptions. Because all four task values are monotone under input-independent quantum channels, the same scalar source bound constrains every Bob--Eve extension compatible with the complete observed behavior. Using a three-setting extension that separates RAC testing from key generation, we obtain two dimension-independent security certificates over this feasible set: lower bounds on the conditional min-entropy and conditional von Neumann entropy, obtained respectively by direct optimization of Eve's key-guessing probability and by prepare-and-measure semidefinite relaxations based on the Brown--Fawzi--Fawzi variational bound. The exclusion-assisted assumptions certify positive key rates down to nearly vanishing preparation visibility, far beyond four-state or parity discrimination alone. Under direct-sum leakage, all four independently optimized rate bounds remain positive at every sampled point with incomplete leakage and vanish only at complete leakage. These results show that robust semi-device-independent security depends not only on what Eve can identify, but also on what she can exclude.
\end{abstract}

\maketitle

\section{Introduction}

Quantum key distribution derives secrecy from physical law, but its guarantee remains conditional on the device model used in the proof. Device-independent quantum key distribution (DI-QKD) removes the need to trust internal device descriptions by certifying security directly from observed Bell-nonlocal correlations \cite{AcinBrunnerGisinMassarPironioScarani2007DIQKD,PironioAcinBrunnerGisinMassarScarani2009DIQKD,MasanesPironioAcin2011CausalDIQKD,VaziraniVidick2014FullyDIQKD,ArnonFriedmanDupuisFawziRennerVidick2018EAT,PollycenoChaturvediRajDieguezPawlowski2025DIQKDMonogamy}. This stronger guarantee comes at a steep experimental cost: practical implementations must certify loophole-free Bell nonlocality in a regime where every loss event, noise process, and finite-data fluctuation must be treated as potentially adversarial \cite{GiustinaEtAl2015LoopholeFreeBell,ShalmEtAl2015LoopholeFreeBell,HensenEtAl2015LoopholeFreeBell,MurtaVanDamRibeiroHansonWehner2019TowardsDIQKD,ZapateroEtAl2023AdvancesDIQKD,MiklinChaturvediBourennanePawlowskiCabello2022DetectionEfficiency,ChaturvediViolaPawlowski2024ArbitraryDistances,GigenaPanwarScalaAraujoFarkasChaturvedi2025TiltedStrategies}. Semi-device-independent quantum cryptography occupies the resulting middle ground: it leaves the measurements uncharacterized while replacing Bell certification with a trusted assumption on the source.

\begin{figure}[t]
\centering
\includegraphics[width=\columnwidth]{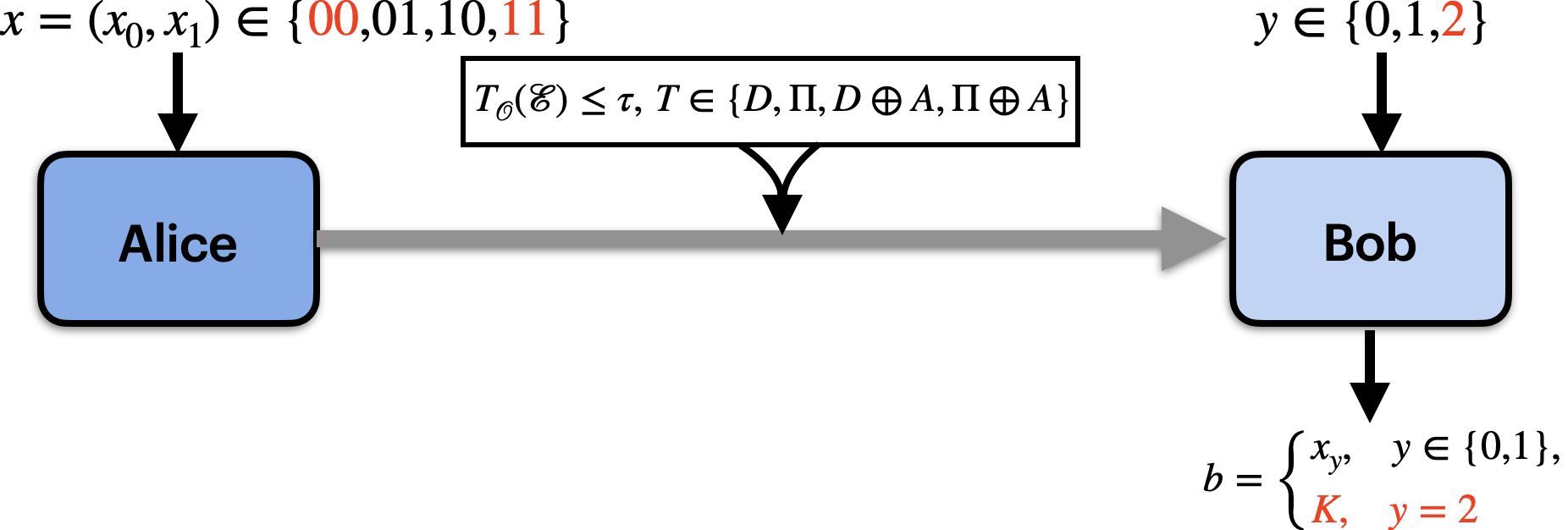}
\caption{Three-setting retained-key protocol under an operational source
assumption. Alice emits one of four preparations satisfying
\(T_{\mathcal O}(\mathcal E)\leq\tau\), with
\(T\in\{D,\Pi,D\oplus A,\Pi\oplus A\}\); in the security analysis
\(\mathcal O=\mathcal Q\). Bob uses \(y=0,1\) for RAC testing, where the
target is \(x_y\), and \(y=2\) for retained-key rounds on the branch
\(x_Z(0)=00\), \(x_Z(1)=11\), where the target is the raw key bit \(K\).}
\label{fig:protocol-schematic}
\end{figure}

The original SDI-QKD construction of Paw{\l}owski and Brunner assumes an upper bound on the Hilbert-space dimension of Alice's transmitted system while leaving Bob's measurements uncharacterized \cite{PawlowskiBrunner2011SDIQKD}. In its two-bit random-access-code protocol, Alice receives $x=(x_0,x_1)$, encodes it into a state supported on a Hilbert space of dimension at most $d$, and Bob chooses $y\in\{0,1\}$ to recover $x_y$. Apart from this dimension bound, neither Alice's preparations nor Bob's measurements are specified. The model has since supported further SDI security proofs, connections with nonlocal resources, and experimental certification protocols \cite{ChaturvediRayVeynarPawlowski2018SDIQKD,ChaturvediPawlowskiHorodecki2017RACNonlocalResources,RajPrasadChaturvediEtAl2026WaveParticleSDI,WoodheadPironio2015PMCHSHQubitBound,LiEtAl2011SDIRNGNoEntanglement,LiPawlowskiYinGuoHan2012SDIRACRandomness,LunghiEtAl2015SelfTestingQRNG,PollycenoFreudenheimNogueiraChaturvediRabeloPawlowski2026CommunicationConstrained}. For the symmetric BB84 RAC behavior $P_B(v)=\tfrac12+\tfrac{v}{2\sqrt2}$, where $v$ is the net visibility of the observed RAC behavior, the original criterion has critical visibility $v_{\rm c}^{\rm orig}=(\sqrt2+\sqrt6)/4\approx0.9659$, while its recent refinement lowers this to $v_{\rm c}^{\rm ref}=2\sqrt2/3\approx0.9428$ \cite{PawlowskiBrunner2011SDIQKD,RajPrasadChaturvediEtAl2026WaveParticleSDI}. Both criteria therefore remain confined to a narrow high-visibility regime. Overcoming these stringent thresholds requires addressing three distinct limitations of the existing framework: the two-setting RAC protocol, the entropy certificate used to bound Eve's information, and the dimension-only description of Alice's source. At the protocol level, Bob's two RAC settings are used both to test the observed correlations and to generate the raw key, so even the ideal BB84 realization retains the intrinsic RAC decoding error. At the security-proof level, the entropy certificate converts these correlations into a bound on Eve's uncertainty and thereby controls the certified noise tolerance. At the source level, a dimension bound constrains the size of the carrier Hilbert space but does not directly quantify the information operationally accessible from the preparation ensemble. Exact finite-dimensional support is also not a natural experimental control for optical sources, whose emitted states occupy an a priori unbounded Fock space and can instead be constrained through mean-energy or mean-photon-number bounds \cite{VanHimbeeckWoodheadCerfGarciaPatronPironio2017NaturalAssumptions,GallegoBrunnerHadleyAcin2010DimensionWitnesses,BrunnerNavascuesVertesi2013DimensionStateDiscrimination,BowlesQuintinoBrunner2014IndependentDevices,PauwelsPironioWoodheadTavakoli2022AlmostQudits,RochCarcellerPauwelsPironioTavakoli2025PhotonNumberConstraints}.

The relevant question is therefore not how large the carrier space is, but how much information Alice's preparation ensemble makes operationally accessible. Physical source models address this question through trusted bounds on physical properties of the emitted states. Beyond energy and photon number, assumptions based on preparation overlaps or fidelities to target states underpin state-discrimination-based randomness generation, measurement certification, receiver-device-independent QKD, and restricted-distrust models \cite{BraskEtAl2017SDIQRNGUSD,ShiCaiBraskZbindenBrunner2019MinimumOverlap,IoannouEtAl2022RDIQKD,IoannouSekatskiAbbottRossetBancalBrunner2022RDIProtocols,Tavakoli2021RestrictedDistrust,BraskBrunnerPauwelsRuscaTavakoli2026PMReview,RochCarcellerPauwelsPironioTavakoli2025PhotonNumberConstraints,TavakoliCruzeiroBraskGisinBrunner2020InformationallyRestricted,TavakoliCruzeiroWoodheadPironio2022InformationallyRestricted}. Despite their different physical content, a broad class of these assumptions has the same operational consequence: an upper bound on the optimal probability with which any receiver can identify the emitted state \cite{PauwelsPironioTavakoli2025InformationCapacity}; for Alice's four-preparation ensemble, we denote this four-state discrimination probability by $D$.

An operational source assumption may concern a function of Alice's string $x$ rather than the full string. For parity, let $\bar\rho_e:=\frac12\sum_{x:\,x_0\oplus x_1=e}\rho_x$ be the uniform mixture of the preparations in parity class $e\in\{0,1\}$, and let $\Pi$ be the optimal probability of distinguishing $\bar\rho_0$ from $\bar\rho_1$. The minimum value $\Pi=\tfrac12$ is exact parity obliviousness: under preparation noncontextuality it yields the classical parity-oblivious RAC bound violated by quantum theory \cite{SpekkensBuzacottKeehnTonerPryde2009POM,ChaillouxKerenidisKunduSikora2016ParityObliviousRAC,AmbainisBanikChaturvediKravchenkoRai2019ParityObliviousDLevel,ChaturvediFarkasWright2021ContextualityScenarios,SahaChaturvedi2019PreparationContextualityCommunication,VaisakhPatraJanpanditSenBanikChaturvedi2021MUBSRAC}, while under the standard Bell-to-RAC relabelling with uniform Alice marginals it is the prepare-and-measure form of no-signalling \cite{WrightFarkas2023BellContextualityMap}. More generally, $\Pi$ is the optimal one-shot probability of inferring the parity of Alice's string $x$, which under this relabelling corresponds to Alice's measurement setting in the Bell scenario.

Four-state and parity discrimination quantify what a receiver can identify. A measurement may fail to identify the emitted state and nevertheless rule out one candidate state. We denote the optimal probability of doing so by $A$, the optimized conclusive-exclusion probability of the preparation ensemble \cite{CavesFuchsSchack2002Compatibility,BandyopadhyayJainOppenheimPerry2014ConclusiveExclusion}. Antidistinguishability underpins the Pusey--Barrett--Rudolph theorem \cite{PuseyBarrettRudolph2012Reality} and has established roles in $\psi$-ontology, contextuality, state geometry, and operational communication \cite{Leifer2014PsiOntologyReview,LeiferDuarte2020Antidistinguishability,SrikumarBartlettKaranjai2026ContextualityAntidistinguishability,JohnstonRussoSikora2025AntidistinguishabilityBounds,PanditHazraMannaChaturvediSaha2026AntiDistConstraints,ChaturvediSaha2020OntologicallyOperationally,MannaChaturvediSaha2024DistinguishabilityComplexity,RayChaturvediSaha2025GamblingTwoQubits,MannaPanditChaturvediSaha2026EntanglementAssistedDistinguishability}. Every four-state discrimination measurement induces an exclusion strategy, yet $D$ does not determine $A$; exclusion therefore captures operational information not fixed by distinguishability alone. For phase-randomized optical encodings with a calibrated common vacuum contribution, the same $x$-independent component constrains identification and exclusion jointly, providing a direct physical route to the exclusion-assisted assumptions considered here \cite{CaoZhangLoMa2015DiscretePhaseWCS}.

We therefore formulate four operational source assumptions, each specified by a single scalar upper bound on $D$, $\Pi$, $D\oplus A=(D+A)/2$, or $\Pi\oplus A=(\Pi+A)/2$. For each assumption, we derive the exact classical RAC frontier without restricting the classical message alphabet and obtain certified arbitrary-dimensional pointwise upper bounds on the quantum deviation. Any frontier violation certifies that, in every compatible quantum realization, Alice's preparations cannot all commute and Bob's measurements must be incompatible \cite{ChaturvediPawlowskiSaha2026EpistemicIncompleteness}. At the reported numerical precision, the BB84 RAC strategy attains the maximal quantum deviation from all four classical frontiers; for each exclusion-assisted assumption, the preparation-depolarized BB84 family, parameterized by the preparation visibility $\nu\in[0,1]$, traces the quantum boundary below the corresponding BB84 task value, while the direct-sum leakage family, parameterized by the leakage fraction $\ell\in[0,1]$, traces it above. The two branches meet at ideal BB84, where $(\nu,\ell)=(1,0)$. Under receiver-side depolarization of the reference decoders used to generate the observed behavior, both exclusion-assisted deviations remain positive for every nonzero preparation visibility exactly when $\nu_{\rm M}>1/\sqrt2$. These complementary boundary families define the two source regimes examined in the security analysis.

These operational source assumptions become security constraints through the three-setting retained-key protocol summarized in Fig.~\ref{fig:protocol-schematic}. Bob uses \(y\in\{0,1\}\) for RAC testing, where the target is \(x_y\). In key rounds, Alice samples \(K\in\{0,1\}\), emits the preparation with \(x=x_Z(K)\), with \(x_Z(0)=00\) and \(x_Z(1)=11\), and Bob uses the additional uncharacterized setting \(y=2\), where the target is \(K\). The ideal BB84 realization has zero key-basis error on this retained branch. By data processing, \(D\), \(\Pi\), \(D\oplus A\), and \(\Pi\oplus A\) are nonincreasing under input-independent quantum channels, so the same scalar source bound applies to every Bob--Eve extension of the emitted ensemble. Security is then certified against all unrestricted-dimensional extensions compatible with the complete three-setting behavior, rather than with the RAC score alone. Over this feasible set, we obtain two dimension-independent security certificates. Direct optimization of Eve's key-guessing probability yields a lower bound on the conditional min-entropy of the raw key given Eve. The second certificate uses prepare-and-measure semidefinite relaxations for conditional entropy based on the Brown--Fawzi--Fawzi variational bound (PM-BFF) to lower-bound the von Neumann entropy term in the Devetak--Winter rate \cite{DevetakWinter2005SecretKeyDistillation,BrownFawziFawzi2021ConditionalEntropies,BrownFawziFawzi2024ConditionalVonNeumannEntropy}. At finite level, the PM-BFF certificate is tightened by enforcing the Eve-side Sylvester stationarity relations and constructing norm localizers from the sharp operator-norm bound on the auxiliary operators.

The numerical comparison isolates two distinct sources of robustness. For a fixed operational source assumption and observed behavior, the PM-BFF conditional-entropy certificate enlarges the positive-rate region beyond direct min-entropy certification. Already under the coarse dimension-implied bound $D\leq\tfrac12$, the three-setting PM-BFF analysis lowers the critical visibility below the earlier RAC-based thresholds \cite{PawlowskiBrunner2011SDIQKD,RajPrasadChaturvediEtAl2026WaveParticleSDI}; using the exact distinguishability $D^{(\nu)}$ of the emitted ensemble lowers it further. A much larger gain comes from the source assumption itself: under preparation depolarization, bounds on $D$ or $\Pi$ certify key only at high source visibility, whereas $D\oplus A$ and $\Pi\oplus A$ remain key-positive down to nearly vanishing $\nu$. This exclusion advantage survives receiver-side depolarization of the reference decoders used to generate the observed behavior. Along the complementary direct-sum leakage family, the retained key branch remains error free while the emitted ensemble increasingly reveals $x$; all four independently optimized rate bounds remain positive at every sampled point with incomplete leakage and vanish only at complete leakage. The hierarchy is therefore structural: PM-BFF strengthens the entropy certificate on a fixed Bob--Eve feasible set, while the exclusion-assisted assumptions remove extensions that identification-only bounds leave admissible.

\section{The random-access code under operational source assumptions}
\label{sec:rac-source}

In a two-bit random-access code, Alice receives two uniformly distributed bits
$x=(x_0,x_1)\in\mathcal X:=\{0,1\}^2$ and implements the preparation $P_x$. Bob chooses which bit to retrieve, $y\in\{0,1\}$, performs the corresponding measurement $M_y$, and outputs $b\in\{0,1\}$, succeeding when $b=x_y$. We write
$\mathcal E=\{P_x\}_{x\in\mathcal X}$ for Alice's four-preparation ensemble and
$\mathbf M=\{M_0,M_1\}$ for Bob's decoding strategy.

An operational theory specifies the allowed preparations and measurements together with the rule assigning outcome probabilities to each preparation--measurement pair. Thus, $P_x$ and $M_y$ determine the conditional probability
$p(b|x,y):=p(b|P_x,M_y)$. With uniformly distributed $x$ and $y$, the average RAC score is
\begin{equation}
\label{eq:RAC-tested}
R(\mathcal E,\mathbf M)
=
\frac{1}{8}
\sum_{x\in\mathcal X}
\sum_{y=0}^{1}
p(b=x_y|x,y).
\end{equation}

We first characterize the classical and quantum strategies that generate the RAC behavior. We then define the optimized information-retrieval tasks and formulate the four operational source assumptions as scalar upper bounds, before showing how distinct physical restrictions on Alice's source imply these bounds.

\subsection{Classical and quantum strategies}

The admissible operational strategies are specified by one of two candidate theories
$\mathcal O\in\{\mathcal C,\mathcal Q\}$: classical theory $\mathcal C$ or quantum theory $\mathcal Q$.

In classical theory $\mathcal C$, each preparation $P_x$ is described by a conditional distribution $q(\lambda|x)$ over a finite but otherwise unrestricted message alphabet $\Lambda$, and each decoding measurement $M_y$ by a response function $\xi(b|y,\lambda)$. The resulting behavior is
\[
p(b|x,y)
=
\sum_{\lambda\in\Lambda}
q(\lambda|x)\,\xi(b|y,\lambda).
\]

In quantum theory $\mathcal Q$, each preparation $P_x$ is represented by a density operator $\rho_x$ on an arbitrary Hilbert space, and each decoding measurement $M_y$ by a binary POVM
$M_y=\{M_{b|y}\}_{b\in\{0,1\}}$, with
$M_{b|y}\succeq0$ and $\sum_bM_{b|y}=\mathbb I$. The resulting behavior is
\[
p(b|x,y)
=
\operatorname{Tr}\!\left(\rho_xM_{b|y}\right).
\]

Equivalently, a classical strategy is a quantum strategy in which all preparation states and measurement effects are diagonal in a common basis. No a priori bound is imposed on either the classical message alphabet or the quantum Hilbert-space dimension.

\subsection{Optimized information-retrieval tasks and operational assumptions}

The RAC score $R(\mathcal E,\mathbf M)$ depends on the decoding measurements actually implemented by Bob. An operational source assumption instead concerns the preparation ensemble $\mathcal E$ alone: it bounds the optimum of a prescribed information-retrieval task over all measurements admitted by $\mathcal O$. Let
$\mathsf M_{\mathcal O}(\Omega)$ denote the set of measurements admitted by $\mathcal O$ with outcome alphabet $\Omega$.

We consider three information-retrieval tasks. Four-state discrimination asks which member of the four-preparation ensemble was emitted, parity discrimination asks for $x_0\oplus x_1$, and state exclusion asks for a value $z\neq x$, i.e. one ensemble member that was not emitted. Their optimal success probabilities are
\begin{align}
D_{\mathcal O}(\mathcal E)
&=
\frac{1}{4}
\sup_{N\in\mathsf M_{\mathcal O}(\mathcal X)}
\sum_{x\in\mathcal X}
p_N(x|P_x),
\label{eq:source-task-definitions}
\\[1mm]
\Pi_{\mathcal O}(\mathcal E)
&=
\frac{1}{4}
\sup_{F\in\mathsf M_{\mathcal O}(\{0,1\})}
\sum_{x\in\mathcal X}
p_F(x_0\oplus x_1|P_x),
\notag
\\[1mm]
A_{\mathcal O}(\mathcal E)
&=
\frac{1}{4}
\sup_{G\in\mathsf M_{\mathcal O}(\mathcal X)}
\sum_{x\in\mathcal X}
\sum_{\substack{z\in\mathcal X\\ z\neq x}}
p_G(z|P_x).
\notag
\end{align}
We call $D_{\mathcal O}$ the four-state discrimination probability,
$\Pi_{\mathcal O}$ the parity discrimination probability, and
$A_{\mathcal O}$ the state-exclusion probability. These optima are linked by admissible post-processings of the measurement outcomes. Coarse-graining a four-state discrimination measurement
$\{N_x\}_{x\in\mathcal X}$ according to parity,
$F_e=\sum_{x:\,x_0\oplus x_1=e}N_x$, gives
$\Pi_{\mathcal O}\geq D_{\mathcal O}$, while splitting each parity outcome uniformly between its two compatible values of $x$,
$N_x=F_{x_0\oplus x_1}/2$, gives
$D_{\mathcal O}\geq\Pi_{\mathcal O}/2$.

The same measurements also induce exclusion strategies. A four-state discrimination measurement gives
$G_x=(\mathbb I-N_x)/3$, while a parity measurement gives
$G_z=F_{1\oplus z_0\oplus z_1}/2$. In the latter construction, a correct parity outcome excludes the reported value of $x$ with certainty, whereas an incorrect parity outcome succeeds with probability $1/2$. Therefore,
\begin{equation}
\label{eq:task-relations}
\begin{aligned}
\frac{\Pi_{\mathcal O}(\mathcal E)}{2}
&\leq
D_{\mathcal O}(\mathcal E)
\leq
\Pi_{\mathcal O}(\mathcal E),
\\[1mm]
A_{\mathcal O}(\mathcal E)
&\geq
\max\!\left\{
\frac{2+D_{\mathcal O}(\mathcal E)}{3},
\frac{1+\Pi_{\mathcal O}(\mathcal E)}{2}
\right\}.
\end{aligned}
\end{equation}

All bounds in Eq.~\eqref{eq:task-relations} are tight. The no-information strategy gives
$(D_{\mathcal C},\Pi_{\mathcal C},A_{\mathcal C})=(1/4,1/2,3/4)$, saturating
$D_{\mathcal C}=\Pi_{\mathcal C}/2$ and both lower bounds on $A_{\mathcal C}$, while transmission of $x$ in full gives
$(1,1,1)$, saturating
$D_{\mathcal C}=\Pi_{\mathcal C}$ and both exclusion bounds. The state-exclusion task also carries information not fixed by the identification optima.

\begin{proposition}[Exclusion is not fixed by identification]
\label{prop:exclusion-not-fixed}
There exist classical ensembles \(\mathcal E\) and \(\mathcal E'\) with
identical values of \(D_{\mathcal C}\) and \(\Pi_{\mathcal C}\), but with
different values of \(A_{\mathcal C}\). Consequently, \(A\) is not
determined by the pair \((D,\Pi)\), even within the classical subset of
quantum theory.
\end{proposition}

\begin{proof}
A flagged mixture that encodes \(x_0\) with probability \(2/3\) and the parity
\(x_0\oplus x_1\) with probability \(1/3\) gives
\((D,\Pi,A)=(1/2,2/3,1)\). A flagged mixture of the no-information strategy
with probability \(2/3\) and the complete-$x$ strategy with probability \(1/3\)
gives \((D,\Pi,A)=(1/2,2/3,5/6)\).
\end{proof}

Proposition~\ref{prop:exclusion-not-fixed} is the structural fact the composite assumptions exploit. A bound on $D\oplus A$ or $\Pi\oplus A$ removes adversarial extensions that no bound on $(D,\Pi)$ can remove, because the state-exclusion task distinguishes ensembles that identification renders indistinguishable. Figure~\ref{fig:task-lattice} summarizes the three tasks, their composites, and the relations of Eq.~\eqref{eq:task-relations}.

\begin{figure}[t]
\centering
\begin{tikzpicture}[>=Stealth,
  task/.style={draw, rounded corners=2pt, inner sep=4pt, minimum width=15mm, minimum height=8mm, align=center, font=\small},
  comp/.style={draw, rounded corners=2pt, inner sep=4pt, minimum width=17mm, minimum height=8mm, align=center, font=\small, fill=black!6},
  lab/.style={font=\scriptsize, fill=white, inner sep=1pt}]
\node[task] (D) at (0,0) {$D$\\[-2pt]\scriptsize 4-state discr.};
\node[task] (P) at (4.6,0) {$\Pi$\\[-2pt]\scriptsize parity};
\node[task] (A) at (2.3,1.7) {$A$\\[-2pt]\scriptsize exclusion};
\node[comp] (DA) at (0,3.4) {$D\oplus A$\\[-2pt]\scriptsize $(D+A)/2$};
\node[comp] (PA) at (4.6,3.4) {$\Pi\oplus A$\\[-2pt]\scriptsize $(\Pi+A)/2$};
\draw[<->] (D) -- node[lab,below=1pt]{$\Pi/2\leq D\leq\Pi$} (P);
\draw[->,dashed] (D) -- node[lab,sloped,above]{$A\geq\frac{2+D}{3}$} (A);
\draw[->,dashed] (P) -- node[lab,sloped,above]{$A\geq\frac{1+\Pi}{2}$} (A);
\draw[->] (D) -- (DA);
\draw[->] (A) -- (DA);
\draw[->] (P) -- (PA);
\draw[->] (A) -- (PA);
\node[lab] at (2.3,-0.95) {$F_D(\tau)=\min\{2\tau,\tfrac{1+\tau}{2}\}$,\quad $F_\Pi(\tau)=\tfrac{1+\tau}{2}$};
\node[lab] at (2.3,4.35) {$F_{D\oplus A}(\tau)=\tau$,\quad $F_{\Pi\oplus A}(\tau)=\min\{2\tau-\tfrac34,\tau\}$};
\end{tikzpicture}
\caption{Operational source tasks and their relations. Boxes are the three
optimized information-retrieval tasks on Alice's four-preparation ensemble;
shaded boxes are the two-setting composites used as exclusion-assisted
source assumptions. Solid double arrow: coarse-graining relations between
four-state and parity discrimination. Dashed arrows: the induced exclusion
strategies of Eq.~\eqref{eq:task-relations}; all displayed relations are
tight. By Proposition~\ref{prop:exclusion-not-fixed}, $A$ is not a function
of $(D,\Pi)$. The annotations give the exact classical RAC frontier of
Theorem~\ref{thm:classical-frontiers} for each source assumption on its
stated domain.}
\label{fig:task-lattice}
\end{figure}
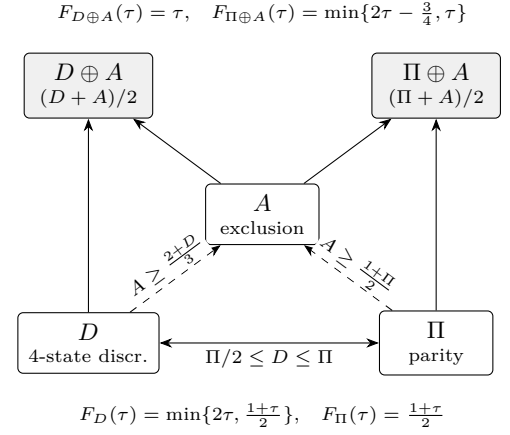

We study four operational source assumptions, each specified by a single scalar upper bound,
\begin{equation}
\label{eq:source-assumptions}
T_{\mathcal O}(\mathcal E)\leq\tau,
\qquad
T\in\{D,\Pi,D\oplus A,\Pi\oplus A\}.
\end{equation}
The exclusion-assisted task values $D\oplus A$ and $\Pi\oplus A$ are themselves optimized two-setting information-retrieval tasks. In each case, a uniformly random setting is revealed before the measurement. The receiver is asked either to discriminate the emitted state or to exclude one candidate state for $D\oplus A$, and either to discriminate the parity or to exclude one candidate state for $\Pi\oplus A$. Because the two setting-dependent measurements are chosen independently, optimizing the average score gives the average of the two optimal scores:
\begin{align}
(D\oplus A)_{\mathcal O}(\mathcal E)
&=
\frac{
D_{\mathcal O}(\mathcal E)
+
A_{\mathcal O}(\mathcal E)
}{2},
\label{eq:composite-source-tasks}
\\
(\Pi\oplus A)_{\mathcal O}(\mathcal E)
&=
\frac{
\Pi_{\mathcal O}(\mathcal E)
+
A_{\mathcal O}(\mathcal E)
}{2}.
\notag
\end{align}
A bound on either composite therefore constrains the optimal average success probability in the corresponding two-setting task; it neither requires a common measurement nor imposes separate bounds on the component optima. 

\subsection{Physical source models underlying the operational assumptions}

A qubit source provides the most direct connection with the original SDI model. If the four preparations are supported on a two-dimensional Hilbert space, then
$0\preceq\rho_x\preceq\mathbb I_2$, and every four-state discrimination POVM satisfies
\begin{equation}
\label{eq:qubit-implies-D-main}
D_{\mathcal Q}(\mathcal E)
=
\frac{1}{4}
\max_{\{N_x\}}
\sum_{x\in\mathcal X}
\operatorname{Tr}(\rho_xN_x)
\leq
\frac{1}{4}
\sum_{x\in\mathcal X}
\operatorname{Tr}N_x
=
\frac{1}{2}.
\end{equation}
The same argument gives
$D_{\mathcal Q}(\mathcal E)\leq d/4$
for preparations supported on a $d$-dimensional Hilbert space. These implications are one-way: finite-dimensional support is a sufficient physical origin of the distinguishability bound, whereas the operational assumption itself places no restriction on the Hilbert-space dimension.

Exact parity obliviousness provides the corresponding route to the parity assumption. The condition
$\rho_{00}+\rho_{11}=\rho_{01}+\rho_{10}$
makes the two effective parity preparations identical and therefore gives
$\Pi_{\mathcal Q}(\mathcal E)=1/2$. Under the standard relabelling
$x=(a,a\oplus s)$, the same identity is the prepare-and-measure form of no-signalling for remotely prepared states with uniform Alice marginals, while
$b=x_y$ becomes the CHSH winning condition
$a\oplus b=sy$ \cite{WrightFarkas2023BellContextualityMap}.

The exclusion-assisted assumptions arise naturally when all four
preparations contain a common $x$-independent component. Let
\(\rho_x=q\sigma+(1-q)\widetilde\rho_x\), where \(\sigma\) is normalized
and independent of $x$. On the common component, the value of $x$ is
completely hidden: the optimal four-state discrimination, parity discrimination, and state-exclusion
success probabilities are \(1/4\), \(1/2\), and \(3/4\), respectively.
Allowing the residual preparations \(\widetilde\rho_x\) to carry
arbitrary information about $x$ gives
\begin{align}
(D\oplus A)_{\mathcal Q}(\mathcal E)
&\leq
1-\frac{q}{2},
\label{eq:common-component-composite-main}
\\
(\Pi\oplus A)_{\mathcal Q}(\mathcal E)
&\leq
1-\frac{3q}{8}.
\notag
\end{align}
These bounds are tight given only the common weight \(q\), since the
residual preparations may encode the four values of $x$ into mutually
orthogonal states.

Phase-randomized optical encodings realize this structure through the
vacuum sector. After trusted complete randomization of the global optical
phase, each preparation is block diagonal in total photon number and
contains the common state
\(\ket{\mathrm{vac}}\!\bra{\mathrm{vac}}\)
\cite{CaoZhangLoMa2015DiscretePhaseWCS}. If \(p_0(x)\) is the vacuum
probability, the common weight can be chosen as \(q=\min_x p_0(x)\).
Hence any lower bound on the vacuum probability is inserted directly
into Eq.~\eqref{eq:common-component-composite-main}. For example, a
uniform non-vacuum bound \(1-p_0(x)\leq\omega\) gives
\((D\oplus A)_{\mathcal Q}\leq(1+\omega)/2\) and
\((\Pi\oplus A)_{\mathcal Q}\leq5/8+3\omega/8\).

A mean-photon-number bound gives a directly measurable sufficient
condition because
\(1-p_0(x)\leq\operatorname{Tr}(\hat n\rho_x)\). Hence
\(\operatorname{Tr}(\hat n\rho_x)\leq\bar n_{\max}\) for every $x$
implies \(q\geq\max\{0,1-\bar n_{\max}\}\)
\cite{VanHimbeeckWoodheadCerfGarciaPatronPironio2017NaturalAssumptions}.
For phase-randomized weak coherent preparations with total emitted mean
photon numbers \(\mu_x\leq\mu_{\max}\), one has
\(p_0(x)=e^{-\mu_x}\), and therefore \(q\geq e^{-\mu_{\max}}\). The same
optical parameter thus constrains identification and exclusion jointly
through Eq.~\eqref{eq:common-component-composite-main}. The phase
randomization is trusted, neither Bob nor Eve has access to the phase or
its record, and the bounds apply to Alice's emitted ensemble before any
untrusted loss-conditioned postselection. Appendix~\ref{app:operational-origins}
gives the complete derivation under the stated source conditions.

Here \emph{microscopic} refers to a theory-specific description of the internal source model, such as dimension, photon number, overlaps, or a calibrated optical decomposition. In contrast, the operational source assumptions used here are theory-independent task-level bounds on the emitted ensemble. Once a microscopic model supplies such a bound
$T_{\mathcal Q}(\mathcal E)\leq\tau$, its remaining details do not enter the analysis. We now determine the largest classical RAC score compatible with that bound.

\section{Exact classical frontiers}
\label{sec:classical-frontiers}

For each
$T\in\{D,\Pi,D\oplus A,\Pi\oplus A\}$,
let $F_T(\tau)$ be the supremum of the RAC score over strategies in
$\mathcal C$ satisfying
$T_{\mathcal C}(\mathcal E)\leq\tau$.
The classical message alphabet is unrestricted.

\begin{theorem}[Exact classical RAC frontiers]
\label{thm:classical-frontiers}
The four frontiers are
\begin{equation}
\label{eq:classical-frontier-functions}
\begin{aligned}
F_D(\tau)
&=
\begin{cases}
2\tau,
& \dfrac14\leq\tau\leq\dfrac13,\\[1mm]
\dfrac{1+\tau}{2},
& \dfrac13\leq\tau\leq1,
\end{cases}
\\[2mm]
F_\Pi(\tau)
&=
\frac{1+\tau}{2},
\qquad
\frac12\leq\tau\leq1,
\\[2mm]
F_{D\oplus A}(\tau)
&=
\tau,
\qquad
\frac12\leq\tau\leq1,
\\[2mm]
F_{\Pi\oplus A}(\tau)
&=
\begin{cases}
2\tau-\dfrac34,
& \dfrac58\leq\tau\leq\dfrac34,\\[1mm]
\tau,
& \dfrac34\leq\tau\leq1.
\end{cases}
\end{aligned}
\end{equation}
Every value is attained throughout the stated domain.
\end{theorem}

\begin{proof}[Proof sketch]
Condition on the classical message $\lambda$ and write
$q_x=p(x|\lambda)$ for the posterior distribution over the four values of $x$.
After independently relabelling the two bit values so that $0$ is the more
likely value of each bit, the optimal conditional RAC score is
$r(q)=\tfrac12+\tfrac12(q_{00}-q_{11})$. For the same posterior, four-state discrimination, parity discrimination, and state exclusion have optimal values
$d(q)=\max_x q_x$, 
$\pi(q)=\max\{q_{00}+q_{11},q_{01}+q_{10}\}$, and
$a(q)=1-\min_x q_x$, respectively.

The opposite-corner contrast $q_{00}-q_{11}$ then gives the pointwise bounds
$r(q)\leq 2d(q)$,
$r(q)\leq[1+d(q)]/2$,
$r(q)\leq[1+\pi(q)]/2$,
$r(q)\leq[d(q)+a(q)]/2$,
$r(q)\leq\pi(q)+a(q)-3/4$, and
$r(q)\leq[\pi(q)+a(q)]/2$.
Because $\lambda$ is available to every decoder, the global RAC score and
source-task values are the corresponding averages of these posterior
quantities. Averaging the affine inequalities and imposing
$T_{\mathcal C}(\mathcal E)\leq\tau$ yields the four expressions in
Eq.~\eqref{eq:classical-frontier-functions}.

The endpoints are attained by elementary encodings carrying no information,
excluding one value of $x$, revealing one bit of $x$, or revealing $x$ in full. Flagged mixtures of the appropriate endpoint encodings attain every
affine segment. Appendix~\ref{app:classical-frontiers} gives the complete
derivation and the explicit attaining strategies.
\end{proof}

These frontiers fix the classical baseline at every admissible source
task value. We next determine the quantum deviation attainable by
arbitrary-dimensional strategies.
\section{Quantum deviation from the classical frontiers}
\label{sec:quantum-deviation}

For a quantum strategy $(\mathcal E,\mathbf M)$ and
$T\in\{D,\Pi,D\oplus A,\Pi\oplus A\}$, define the signed deviation from
the corresponding exact classical frontier in
Eq.~\eqref{eq:classical-frontier-functions} by
\begin{equation}
\label{eq:quantum-gap-functional}
\Delta_T(\mathcal E,\mathbf M)
=
R(\mathcal E,\mathbf M)
-
F_T\!\left(T_{\mathcal Q}(\mathcal E)\right).
\end{equation}
Here $R(\mathcal E,\mathbf M)$ is evaluated using Bob's implemented
decoding measurements, whereas $T_{\mathcal Q}(\mathcal E)$ is optimized
over all quantum measurements on Alice's ensemble. Thus $\Delta_T>0$
certifies a violation of the corresponding classical frontier and
simultaneously witnesses that Alice's preparations are not jointly
commuting and that Bob's implemented decoding measurements are
incompatible
\cite{ChaturvediPawlowskiSaha2026EpistemicIncompleteness}.

To distinguish the pointwise quantum boundary from its global extremum,
define
\begin{equation}
\label{eq:pointwise-maximum-deviation}
\Delta_T^{\max}(t)
:=
\sup_{\substack{\mathcal E,\mathbf M\\
T_{\mathcal Q}(\mathcal E)=t}}
\Delta_T(\mathcal E,\mathbf M),
\qquad
\Delta_T^\star
:=
\sup_t\Delta_T^{\max}(t).
\end{equation}
Here $\Delta_T^{\max}(t)$ is the pointwise quantum excess above the exact
classical frontier at source value $t$, while $\Delta_T^\star$ is the
largest separation between the arbitrary-dimensional quantum region and
that frontier.

\subsection{Explicit deviations and thresholds}
\label{subsec:explicit-deviations-thresholds}

We now compute the feasible quantum deviations for two deformations of
the BB84 ensemble. One removes state contrast by depolarizing Alice's
preparations; the other adds an orthogonal $x$-revealing sector. Let
\begin{equation}
\label{eq:bb84-source}
\rho_x^{\rm BB84}
=
\frac12
\left[
\mathbb I_2
+
\frac{(-1)^{x_0}X+(-1)^{x_1}Z}{\sqrt2}
\right]
\end{equation}
be the ideal BB84 RAC ensemble, and define
\begin{equation}
\label{eq:bb84-source-families}
\begin{aligned}
\rho_x^{(\nu)}
&=
\nu\rho_x^{\rm BB84}
+
(1-\nu)\frac{\mathbb I_2}{2},
\\
\rho_x^{(\ell)}
&=
(1-\ell)\rho_x^{\rm BB84}
\oplus
\ell\ket{x}\!\bra{x},
\end{aligned}
\qquad
\nu,\ell\in[0,1].
\end{equation}
In the second family, the states
$\{\ket{x}\}_{x\in\mathcal X}$ are mutually orthogonal and supported on
a sector orthogonal to the BB84 qubit subspace. The parameter $\nu$
runs from an $x$-independent ensemble to ideal BB84,
whereas $\ell$ runs from ideal BB84 to complete leakage through an orthogonal sector. The two families meet at
$(\nu,\ell)=(1,0)$.

\begin{proposition}[Exact thresholds for the BB84 deformations]
\label{prop:exact-bb84-deformation-thresholds}
For the preparation-depolarized family in
Eq.~\eqref{eq:bb84-source-families}, with Bob using the $X$ decoder for
$y=0$ and the $Z$ decoder for $y=1$, the critical source visibilities
for violating the corresponding classical frontiers are given in
Table~\ref{tab:exact-source-thresholds}. For the direct-sum
leakage family, with Bob using the BB84 decoders on the qubit
sector and reading the requested bit on the orthogonal $x$-revealing sector, all
four deviations are
\begin{equation}
\label{eq:exact-leakage-threshold-summary}
\Delta_T^{(\ell)}
=
(1-\ell)\Delta_0,
\qquad
T\in\{D,\Pi,D\oplus A,\Pi\oplus A\},
\end{equation}
where
\begin{equation}
\label{eq:bb84-common-deviation}
\Delta_0
=
\frac{\sqrt2-1}{4}.
\end{equation}
Thus every point with incomplete leakage $\ell<1$ violates the corresponding
classical frontier, and the violation vanishes only at complete value of $x$
revelation.
\end{proposition}

\begin{table}[t]
\caption{Exact source-visibility thresholds for classical-frontier violation
along the preparation-depolarized BB84 family. The visibility \(\nu\)
changes Alice's emitted ensemble; Bob uses the ideal RAC decoders.}
\label{tab:exact-source-thresholds}
\begin{ruledtabular}
\begin{tabular}{c c}
source assumption & violation condition\\
\hline
$D^{(\nu)}$ & $\nu>(2\sqrt2-1)^{-1}\approx0.5469$\\
$\Pi=1/2$ & $\nu>1/\sqrt2\approx0.7071$\\
$(D\oplus A)^{(\nu)}$ & $\nu>0$\\
$(\Pi\oplus A)^{(\nu)}$ & $\nu>0$\\
$D_{\mathcal Q}\leq1/2$ & $\nu>1/\sqrt2\approx0.7071$\\
\end{tabular}
\end{ruledtabular}
\end{table}

\begin{proof}
At the BB84 endpoint, Bob's $X$ and $Z$ decoders give
\begin{equation}
\label{eq:bb84-task-values}
\begin{aligned}
R_{\rm BB84}
&=
\frac{2+\sqrt2}{4},
&
D_{\mathcal Q}
=
\Pi_{\mathcal Q}
&=
\frac12,
\\
A_{\mathcal Q}
&=
1,
&
(D\oplus A)_{\mathcal Q}
=
(\Pi\oplus A)_{\mathcal Q}
&=
\frac34 .
\end{aligned}
\end{equation}
Each exact classical frontier equals $3/4$ at the corresponding source
task value. Thus all four deviations coincide at
\begin{equation}
\Delta_0
=
R_{\rm BB84}-\frac34
=
\frac{\sqrt2-1}{4}.
\end{equation}

For the preparation-depolarized family, the same decoders give
\begin{equation}
\label{eq:bb84-visibility-task-values}
\begin{aligned}
R^{(\nu)}
&=
\frac12+\frac{\nu}{2\sqrt2},
&
D^{(\nu)}
&=
\frac{1+\nu}{4},
\\
\Pi^{(\nu)}
&=
\frac12,
&
A^{(\nu)}
&=
\frac{3+\nu}{4},
\\
(D\oplus A)^{(\nu)}
&=
\frac12+\frac{\nu}{4},
&
(\Pi\oplus A)^{(\nu)}
&=
\frac58+\frac{\nu}{8}.
\end{aligned}
\end{equation}
Substitution into the exact classical frontiers gives
\begin{equation}
\label{eq:visibility-nonclassicality-deviations}
\begin{aligned}
\Delta_D^{(\nu)}
&=
\begin{cases}
\displaystyle
\frac{\nu}{2}
\left(
\frac{1}{\sqrt2}-1
\right),
&
0\leq\nu\leq\frac13,
\\[2mm]
\displaystyle
\frac{\nu(2\sqrt2-1)-1}{8},
&
\frac13\leq\nu\leq1,
\end{cases}
\\[2mm]
\Delta_\Pi^{(\nu)}
&=
\frac{\sqrt2\,\nu-1}{4},
\\
\Delta_{D\oplus A}^{(\nu)}
&=
\Delta_{\Pi\oplus A}^{(\nu)}
=
\nu\Delta_0.
\end{aligned}
\end{equation}
The thresholds in Table~\ref{tab:exact-source-thresholds} follow
immediately. The weaker qubit implication fixes the relevant classical
value at the same point as the parity case, giving the threshold
$1/\sqrt2$.

For the direct-sum leakage family, Bob obtains the BB84 score on
the qubit sector and perfect decoding on the orthogonal $x$-revealing sector.
The task values are
\begin{equation}
\label{eq:bb84-leakage-task-values}
\begin{aligned}
R^{(\ell)}
&=
R_{\rm BB84}
+
\ell(1-R_{\rm BB84}),
\\
D^{(\ell)}
=
\Pi^{(\ell)}
&=
\frac12+\frac{\ell}{2},
\\
A^{(\ell)}
&=
1,
\\
(D\oplus A)^{(\ell)}
=
(\Pi\oplus A)^{(\ell)}
&=
\frac34+\frac{\ell}{4}.
\end{aligned}
\end{equation}
Substitution into the exact frontiers gives
Eq.~\eqref{eq:exact-leakage-threshold-summary}.
\end{proof}

\subsection{Receiver-side robustness}
\label{subsec:receiver-side-robustness}

The parameter $\nu$ changes Alice's emitted ensemble and hence the
operational task values. Receiver-side depolarization changes
only the reference decoders used to generate the observed RAC behavior.
To separate the two effects, let
\begin{equation}
\label{eq:noisy-bb84-decoders}
M_{b|y}^{(\nu_M)}
=
\nu_M M_{b|y}^{\rm BB84}
+
(1-\nu_M)\frac{\mathbb I_2}{2},
\qquad
y\in\{0,1\},
\end{equation}
where $\nu_M\in[0,1]$ is the receiver-side visibility parameter.

\begin{proposition}[Receiver-visibility thresholds]
\label{prop:receiver-visibility-thresholds}
For the preparation-depolarized source
$\{\rho_x^{(\nu)}\}$ and the depolarized reference decoders in
Eq.~\eqref{eq:noisy-bb84-decoders}, no point on the family violates any
of the four classical frontiers when $\nu_M\leq1/\sqrt2$. When
$\nu_M>1/\sqrt2$, the exact conditions for positive frontier violation
are those in Table~\ref{tab:receiver-visibility-thresholds}.
\end{proposition}

\begin{table}[t]
\caption{Exact frontier-violation conditions for the preparation-depolarized
BB84 source when the reference RAC decoders have receiver-side visibility
\(\nu_M\). The source-task values remain those of Alice's emitted ensemble;
only the observed RAC score is depolarized.}
\label{tab:receiver-visibility-thresholds}
\begin{ruledtabular}
\begin{tabular}{c c}
source assumption & violation condition\\
\hline
$D^{(\nu)}$
&
$\nu_M>1/\sqrt2,\quad
\nu>(2\sqrt2\,\nu_M-1)^{-1}$
\\[1mm]
$\Pi=1/2$
&
$\nu_M>1/\sqrt2,\quad
\nu>(\sqrt2\,\nu_M)^{-1}$
\\[1mm]
$(D\oplus A)^{(\nu)}$
&
$\nu_M>1/\sqrt2,\quad \nu>0$
\\[1mm]
$(\Pi\oplus A)^{(\nu)}$
&
$\nu_M>1/\sqrt2,\quad \nu>0$
\\[1mm]
$D_{\mathcal Q}\leq1/2$
&
$\nu_M>1/\sqrt2,\quad
\nu>(\sqrt2\,\nu_M)^{-1}$
\end{tabular}
\end{ruledtabular}
\end{table}

\begin{proof}
Receiver-side depolarization leaves Alice's ensemble unchanged. Hence
the operational task values remain those in
Eq.~\eqref{eq:bb84-visibility-task-values}. Only the RAC score changes,
and for the depolarized reference decoders in
Eq.~\eqref{eq:noisy-bb84-decoders} it is
\begin{equation}
\label{eq:joint-source-receiver-rac-score}
R^{(\nu,\nu_M)}
=
\frac12+\frac{\nu\nu_M}{2\sqrt2}.
\end{equation}
Substitution into the exact classical frontiers gives
\begin{equation}
\label{eq:joint-source-receiver-deviations}
\begin{aligned}
\Delta_D^{(\nu,\nu_M)}
&=
\begin{cases}
\displaystyle
\frac{\nu}{2}
\left(
\frac{\nu_M}{\sqrt2}-1
\right),
&
0\leq\nu\leq\frac13,
\\[2mm]
\displaystyle
\frac{\nu(2\sqrt2\,\nu_M-1)-1}{8},
&
\frac13\leq\nu\leq1,
\end{cases}
\\[2mm]
\Delta_\Pi^{(\nu,\nu_M)}
&=
\frac{\sqrt2\,\nu\nu_M-1}{4},
\\
\Delta_{D\oplus A}^{(\nu,\nu_M)}
&=
\Delta_{\Pi\oplus A}^{(\nu,\nu_M)}
=
\frac{\nu(\sqrt2\,\nu_M-1)}{4}.
\end{aligned}
\end{equation}
If $\nu_M\leq1/\sqrt2$, all three displayed expressions are
nonpositive for every $\nu\in[0,1]$. If $\nu_M>1/\sqrt2$, solving
Eq.~\eqref{eq:joint-source-receiver-deviations} for positivity gives
the conditions in Table~\ref{tab:receiver-visibility-thresholds}. The
last row follows by replacing the exact task value $D^{(\nu)}$ with the
weaker qubit implication $D_{\mathcal Q}\leq1/2$, which fixes the same
classical value as in the parity case.
\end{proof}

The common threshold $\nu_M=1/\sqrt2$ is the joint-measurability
threshold of the depolarized $X$ and $Z$ decoders
\cite{Busch1986JointMeasurability,HeinosaariMiyaderaZiman2016Incompatibility}.
At or below this value, the implemented decoders are compatible, and no
classical-frontier violation is possible for any source visibility.
Thus Eq.~\eqref{eq:joint-source-receiver-deviations} separates the two
resources required for frontier violation: source contrast through
$\nu$ and measurement incompatibility through $\nu_M$.

\subsection{Bounding the quantum deviation}
\label{subsec:bounding-quantum-deviation}

The constructions above give feasible frontier violations. We now bound
the largest possible deviation over arbitrary prepare-and-measure
realizations. Since $\Delta_T$ contains the optimized source-task value
$T_{\mathcal Q}(\mathcal E)$, we replace each source task by its exact
semidefinite-program dual. The resulting dual constraints restrict the
otherwise unrestricted preparations, and the dual objective supplies the
source-task value in the frontier term.

Writing
$\bar\rho=\frac14\sum_{x\in\mathcal X}\rho_x$, strong duality yields
\begin{equation}
\label{eq:source-tasks-main}
\begin{aligned}
D_{\mathcal Q}(\mathcal E)
&=
\min_{\sigma_D}
\left\{
\Tr\sigma_D:
\sigma_D\succeq\frac14\rho_x
\quad\forall x
\right\},
\\[1mm]
\Pi_{\mathcal Q}(\mathcal E)
&=
\min_{\sigma_\Pi}
\left\{
\Tr\sigma_\Pi:
\begin{array}{l}
\sigma_\Pi\succeq
\frac14(\rho_{00}+\rho_{11}),
\\[-1mm]
\sigma_\Pi\succeq
\frac14(\rho_{01}+\rho_{10})
\end{array}
\right\},
\\[1mm]
A_{\mathcal Q}(\mathcal E)
&=
\min_{\sigma_A}
\left\{
\Tr\sigma_A:
\sigma_A\succeq
\bar\rho-\frac14\rho_z
\quad\forall z
\right\}.
\end{aligned}
\end{equation}
These dual programs give a uniform representation of the four source
task values. We write
$\boldsymbol{\sigma}_D=\sigma_D$,
$\boldsymbol{\sigma}_\Pi=\sigma_\Pi$,
$\boldsymbol{\sigma}_{D\oplus A}=(\sigma_D,\sigma_A)$, and
$\boldsymbol{\sigma}_{\Pi\oplus A}=(\sigma_\Pi,\sigma_A)$.
The associated scalar objectives are
\begin{equation}
\label{eq:source-dual-objectives}
\begin{aligned}
s_D
&=
\Tr\sigma_D,
\\
s_\Pi
&=
\Tr\sigma_\Pi,
\\
s_{D\oplus A}
&=
\frac12\Tr(\sigma_D+\sigma_A),
\\
s_{\Pi\oplus A}
&=
\frac12\Tr(\sigma_\Pi+\sigma_A).
\end{aligned}
\end{equation}
Let $\mathfrak D_T(\mathcal E)$ denote the feasible set defined by the
dual inequalities for the source assumption $T$, with the Cartesian
product understood for the composite assumptions. Then
\begin{equation}
\label{eq:unified-source-dual}
T_{\mathcal Q}(\mathcal E)
=
\min_{\boldsymbol{\sigma}_T\in\mathfrak D_T(\mathcal E)}
s_T.
\end{equation}

Because Bob's measurement device and Hilbert-space dimension are
unrestricted, a common Naimark dilation allows both binary decoders to
be taken projective,
\begin{equation}
\label{eq:rac-projective-relations}
M_{b|y}M_{b'|y}
=
\delta_{bb'}M_{b|y},
\qquad
\sum_{b=0}^{1}M_{b|y}
=
\mathbb I.
\end{equation}

Let $F_T(s)=c_{T,j}s+d_{T,j}$ on the affine branch $I_{T,j}$.
Substituting Eq.~\eqref{eq:unified-source-dual} into the global
optimization, we maximize
$R-c_{T,j}s_T-d_{T,j}$ over arbitrary-dimensional preparations,
projective decoders, and dual-feasible variables, subject to
$s_T\in I_{T,j}$. This reformulation is exact. Dual feasibility implies
$s_T\geq T_{\mathcal Q}(\mathcal E)$, an optimal dual tuple attains
equality, and the monotonicity of $F_T$ ensures that a nonoptimal tuple
cannot increase the objective. Maximizing over the affine branches
therefore recovers $\Delta_T^\star$.

At fixed source value $t$, imposing $s_T=t$ instead gives an outer
optimization for $\Delta_T^{\max}(t)$. Every ensemble with
$T_{\mathcal Q}(\mathcal E)=t$ is included through an optimal dual
tuple, while a nonoptimal tuple may additionally admit an ensemble with
$T_{\mathcal Q}(\mathcal E)<t$. The resulting optimum is therefore an
upper bound rather than an exact reformulation.

We relax both operator optimizations with the degree-three
operator-family moment relaxation $Q_3$
\cite{TavakoliPozasKerstjensBrownAraujo2024SDPReview}. The relaxation keeps
reduced words through degree three, imposes positivity of the moment
matrix and of the localizing matrices associated with the dual
inequalities. Since every quantum realization defines a feasible
truncated moment sequence,
\begin{equation}
\label{eq:q3-deviation-upper-bounds}
\begin{aligned}
\Delta_T^\star
&\leq
\Delta_{T,Q_3}^{\rm ub},
\\
\Delta_T^{\max}(t)
&\leq
\Delta_{T,Q_3}^{\rm ub}(t).
\end{aligned}
\end{equation}
Finite-dimensional see-saw optimization supplies feasible lower bounds,
and every reported source-task value is recomputed by its exact
semidefinite program. Appendix~\ref{app:quantum-relaxation} gives the
operator programs, moment construction, and numerical validation.
The ideal BB84 deviation from
Proposition~\ref{prop:exact-bb84-deformation-thresholds} can now be
compared with the global arbitrary-dimensional bound.

\begin{result}[Maximum quantum deviation]
\label{res:maximum-quantum-deviation}
For every
$T\in\{D,\Pi,D\oplus A,\Pi\oplus A\}$,
\begin{equation}
\label{eq:maximum-quantum-deviation}
0
\leq
\Delta_T^\star-\Delta_0
\leq
\Delta_{T,Q_3}^{\rm ub}-\Delta_0
<
5\times10^{-6}.
\end{equation}
The BB84 deviation is therefore within $5\times10^{-6}$ of the
arbitrary-dimensional optimum for all four source assumptions.
\end{result}

Figure~\ref{fig:quantum-deviation} compares the feasible deviations from
Proposition~\ref{prop:exact-bb84-deformation-thresholds} with the
pointwise arbitrary-dimensional upper bounds.

\begin{figure*}[t]
\centering
\includegraphics[width=\textwidth]{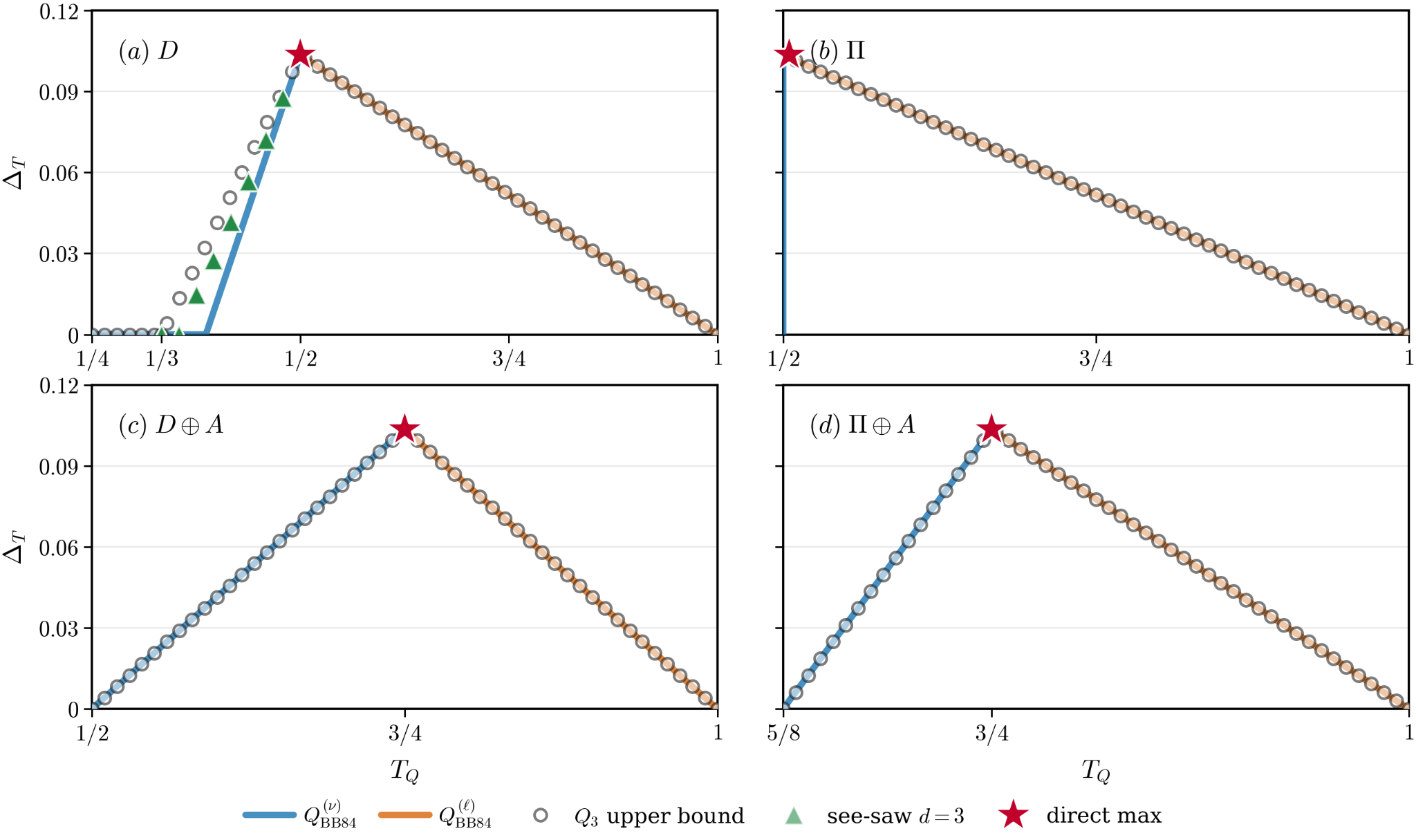}
\caption{Quantum deviation from the exact classical frontiers. Gray circles are
level-\(Q_3\) upper bounds on \(\Delta_T^{\max}(t)\). Blue and orange
curves are feasible preparation-depolarized and direct-sum leakage
families; red stars mark ideal BB84. Green triangles in panel~(a) are
feasible \(d=3\) see-saw points, with the source-task value recomputed by
its exact semidefinite program for the final ensemble.}
\label{fig:quantum-deviation}
\end{figure*}

\begin{result}[Exclusion-assisted quantum boundary]
\label{res:exclusion-assisted-quantum-boundary}
For
$T\in\{D\oplus A,\Pi\oplus A\}$, the preparation-depolarized family
matches the pointwise $Q_3$ upper bound at the sampled source values
below the BB84 task value $t=3/4$, while the direct-sum leakage
family matches it at the sampled source values above $t=3/4$, in both
cases within numerical precision. The two families therefore determine
the sampled arbitrary-dimensional quantum boundary for the
exclusion-assisted source assumptions.
\end{result}

For $\Pi$, preparation depolarization fixes the source-task value at
$\Pi=1/2$, and hence fixes the classical frontier at
$F_\Pi(1/2)=3/4$, while varying the RAC score. This produces the
vertical feasible segment in Fig.~\ref{fig:quantum-deviation};
direct-sum leakage spans the remaining source interval. For $D$,
feasible $d=3$ strategies exceed the qubit depolarization curve over
part of the intermediate range. The remaining gap to the $Q_3$ upper
bound leaves the exact arbitrary-dimensional $D$ boundary unresolved
there.

The preparation-depolarized and direct-sum leakage constructions
therefore provide the two source regimes carried into the security
analysis: loss of source contrast and explicit leakage. The receiver-side depolarization calculation separates these source
effects from measurement incompatibility, while the $Q_3$ bounds show that the
exclusion-assisted regimes lie on the sampled arbitrary-dimensional
quantum boundary. The security analysis now asks the operationally
different question: for the complete three-setting behavior and the same
source constraints, how much information can Eve have about the key?
\section{SDI-QKD under operational source assumptions}
\label{sec:sdiqkd-operational-source}

The previous section used operational task values to compare RAC
scores with exact classical frontiers. We now use the same
four-preparation ensemble for key distribution. Bob keeps the two RAC
settings for testing and adds a third binary setting for key generation.
This separates the RAC test from the key basis while leaving Bob's
measurements uncharacterized. Security is obtained by bounding Eve's
information about the retained raw key from the complete three-setting
behavior and the operational source assumption.

\subsection{Three-setting protocol and key state}
\label{subsec:three-setting-key-state}

The original RAC-based SDI-QKD construction uses one of the two decoding
measurements for key generation \cite{PawlowskiBrunner2011SDIQKD}. Even for the
ideal BB84 RAC strategy, this choice retains the intrinsic RAC error
$1-R_{\rm BB84}=(2-\sqrt2)/4$. We separate RAC testing from key
generation without changing Alice's four-preparation ensemble. Bob uses
$y\in\{0,1\}$ for the RAC test and adds a third binary setting, $y=2$, matched
to the branch
\begin{equation}
\label{eq:key-branch}
x_Z(0)=00,
\qquad
x_Z(1)=11 .
\end{equation}
In the ideal BB84 realization, the two preparations on this branch are
orthogonal and admit an error-free key measurement. This ideal
description is not a trusted-device assumption: all three of Bob's
settings remain uncharacterized. The source assumption applies to the
full emitted state $\rho_x$, including any $x$-dependent side system
or classical flag emitted with the signal.

For each $x\in\mathcal X$, let $\varrho_x^{BE}\succeq0$ be a normalized
joint state, with system $B$ delivered to Bob and system $E$ retained by
Eve. The Hilbert spaces of $B$ and $E$ are unrestricted. Bob's
outcome-zero effect for setting $y\in\{0,1,2\}$ is denoted by $B_y$,
with
\[
B_{0|y}=B_y,
\qquad
B_{1|y}=I_B-B_y .
\]
Since Bob's Hilbert space is unrestricted in the security optimization,
a common Naimark dilation allows the three binary measurements to be
represented projectively,
\begin{equation}
\label{eq:security-projective-bob}
B_y=B_y^\dagger=B_y^2,
\qquad
y\in\{0,1,2\},
\end{equation}
without imposing any commutation relation between distinct settings.
We record the outcome-zero probabilities as the complete response table
\begin{equation}
\label{eq:security-behavior}
\mathbf p
=
\left\{
p_{xy}
:=
\Tr\!\left[
\varrho_x^{BE}
(B_y\otimes I_E)
\right]
\right\}_{x\in\mathcal X,\;y\in\{0,1,2\}} .
\end{equation}
The first eight entries determine the RAC value used in the
classical-frontier test. The security analysis retains all twelve
entries, including the responses of the two non-key preparations to the
key setting. Thus $R$ is the test statistic, whereas $\mathbf p$ is the
data constraining the adversarial realization.

In a retained key round, Alice samples $K\in\{0,1\}$ uniformly and emits
$\varrho_{x_Z(K)}^{BE}$. Bob applies setting $y=2$ and records the
classical outcome $K_B=b$. With the projective representation in
Eq.~\eqref{eq:security-projective-bob}, define
\begin{equation}
\label{eq:eve-subnormalized-key-states}
\tau_{b|k}^{E}
:=
\operatorname{Tr}_{B}\!\left[
(B_{b|2}\otimes I_E)
\varrho_{x_Z(k)}^{BE}
\right],
\qquad
b,k\in\{0,1\}.
\end{equation}
The retained round is described by the classical-classical-quantum state
\begin{equation}
\label{eq:security-ccq-state}
\omega_{K K_B E}
=
\frac{1}{2}
\sum_{k,b=0}^{1}
\lvert k\rangle\!\langle k\rvert_{K}
\otimes
\lvert b\rangle\!\langle b\rvert_{K_B}
\otimes
\tau_{b|k}^{E}.
\end{equation}

\begin{proposition}[Retained key state and rate]
\label{prop:retained-key-state-rate}
Every Bob--Eve realization reproducing
Eq.~\eqref{eq:security-behavior} induces a normalized retained-round
state $\omega_{K K_B E}$ as in Eq.~\eqref{eq:security-ccq-state}. Its
classical marginal gives
\begin{equation}
\label{eq:key-qber}
Q_Z
=
\Pr[K_B\neq K]
=
\frac{1}{2}
\left(
1-p_{00,2}+p_{11,2}
\right),
\end{equation}
and its $K E$ marginal has the form
\begin{equation}
\label{eq:pm-key-reduction}
\omega_{K E}
=
\operatorname{Tr}_{K_B}
\omega_{K K_B E}
=
\frac{1}{2}
\sum_{k=0}^{1}
\lvert k\rangle\!\langle k\rvert_{K}
\otimes
\rho_k^E,
\end{equation}
where
\begin{equation}
\label{eq:eve-key-branch-states}
\rho_k^E
=
\operatorname{Tr}_{B}
\varrho_{x_Z(k)}^{BE},
\qquad
k\in\{0,1\}.
\end{equation}
Consequently, direct reconciliation satisfies
\begin{equation}
\label{eq:dw-key-rate}
r
\ge
H(K|E)_\omega
-
H(K|K_B)_\omega
\ge
H(K|E)_\omega
-
h_2(Q_Z).
\end{equation}
\end{proposition}

\begin{proof}
Taking the \(K E\) marginal of Eq.~\eqref{eq:security-ccq-state} gives
\[
\operatorname{Tr}_{K_B}\omega_{K K_B E}
=
\frac{1}{2}
\sum_{k=0}^{1}
\lvert k\rangle\!\langle k\rvert_K
\otimes
\sum_{b=0}^{1}\tau_{b|k}^{E}.
\]
By the definition of \(\tau_{b|k}^{E}\) and
\(\sum_b B_{b|2}=I_B\),
\[
\sum_{b=0}^{1}\tau_{b|k}^{E}
=
\operatorname{Tr}_{B}\varrho_{x_Z(k)}^{BE}
=
\rho_k^E .
\]
Hence the \(K E\) marginal is Eq.~\eqref{eq:pm-key-reduction}. Since
each \(\varrho_{x_Z(k)}^{BE}\) is normalized, the same identity also
shows that \(\omega_{K K_B E}\) is normalized.

The classical marginal of Eq.~\eqref{eq:security-ccq-state} satisfies
\[
\Pr[K=k,K_B=b]
=
\frac{1}{2}\operatorname{Tr}\tau_{b|k}^{E}.
\]
Therefore
\[
\Pr[K_B\neq K]
=
\frac{1}{2}\operatorname{Tr}\tau_{1|0}^{E}
+
\frac{1}{2}\operatorname{Tr}\tau_{0|1}^{E}
=
\frac{1}{2}
\left(
1-p_{00,2}+p_{11,2}
\right),
\]
where \(x_Z(0)=00\), \(x_Z(1)=11\), \(B_{0|2}=B_2\), and
\(B_{1|2}=I_B-B_2\). This proves Eq.~\eqref{eq:key-qber}. The first
inequality in Eq.~\eqref{eq:dw-key-rate} is the Devetak--Winter
direct-reconciliation bound~\cite{DevetakWinter2005SecretKeyDistillation}; the second is
Fano's bound for a binary raw key.
\end{proof}

Thus the security optimization is a constrained minimization of
$H(K|E)_\omega$: the states $\rho_0^E$ and $\rho_1^E$ are the retained-key
marginals of the same Bob--Eve realization that reproduces the full
response table $\mathbf p$ and satisfies the operational source
assumption.
\subsection{Lifting source assumptions to Bob--Eve extensions}
\label{subsec:source-lift-bobeve}

The source assumption is a constraint on the emitted ensemble
\(\mathcal E=\{\rho_x^A\}_{x\in\mathcal X}\). An admissible
Bob--Eve realization is obtained from these emitted states by a single
quantum post-processing map, independent of the classical value $x$:
\begin{equation}
\label{eq:channel-lift-security}
\varrho_x^{BE}
=
\Lambda_{A\to BE}(\rho_x^A),
\qquad
x\in\mathcal X .
\end{equation}
The map \(\Lambda_{A\to BE}\) is otherwise arbitrary; in particular,
the dimensions of \(B\) and \(E\) are unrestricted. This is precisely
the assumption that any information about $x$ available to Bob or Eve must
come through the emitted state itself, not through an additional
$x$-dependent post-processing.

\begin{proposition}[Data processing for operational source bounds]
\label{prop:channel-lift-source-assumptions}
Let \(T\in\{D,\Pi,D\oplus A,\Pi\oplus A\}\). If
\(T_{\mathcal Q}(\mathcal E)\leq\tau\), then every ensemble
\(\{\varrho_x^{BE}\}_{x\in\mathcal X}\) of the form
Eq.~\eqref{eq:channel-lift-security} satisfies
\begin{equation}
\label{eq:source-monotonicity-security}
T_{\mathcal Q}\!\left(\{\varrho_x^{BE}\}_{x\in\mathcal X}\right)
\leq
\tau .
\end{equation}
Consequently, every physical Bob--Eve realization is feasible for the
lifted security optimization defined by
Eq.~\eqref{eq:source-monotonicity-security}.
\end{proposition}

\begin{proof}
Let \(M\) be any POVM used in one of the guessing or exclusion tasks on
\(\{\varrho_x^{BE}\}_x\). Pulling it back through the adjoint channel
\(\Lambda_{A\to BE}^{\dagger}\) gives a POVM on Alice's emitted system,
because \(\Lambda_{A\to BE}^{\dagger}\) is unital and positive. The
success probability of \(M\) on \(\{\varrho_x^{BE}\}_x\) is exactly the
success probability of the pulled-back POVM on \(\mathcal E\). Thus no
four-state discrimination, parity discrimination, or state-exclusion value can
increase under Eq.~\eqref{eq:channel-lift-security}. Applying this to
the tasks defining \(D\), \(\Pi\), \(D\oplus A\), and \(\Pi\oplus A\)
gives
\[
T_{\mathcal Q}\!\left(\{\varrho_x^{BE}\}_x\right)
\leq
T_{\mathcal Q}(\mathcal E)
\leq
\tau .
\]
For a worst-case lower bound on privacy, enlarging the feasible set can
only decrease the infimum of \(H(K|E)_\omega\), so the resulting key-rate
bound remains valid.
\end{proof}

We use the lifted bound as the source constraint in the adversarial
optimization. Let \(\mathcal F_T(\mathbf p,\tau)\) be the set of all
unrestricted-dimensional Bob--Eve ensembles
\(\{\varrho_x^{BE}\}_{x\in\mathcal X}\) and projective Bob effects
\(\{B_y\}_{y=0}^{2}\) that reproduce
Eq.~\eqref{eq:security-behavior} and satisfy
\[
T_{\mathcal Q}\!\left(\{\varrho_x^{BE}\}_{x\in\mathcal X}\right)
\leq
\tau .
\]
For \(T=D\oplus A\) and \(T=\Pi\oplus A\), this is one scalar constraint
on the composite quantity, not separate constraints on its components.
Define
\begin{equation}
\label{eq:worst-case-privacy}
\mathsf H_T(\mathbf p,\tau)
:=
\inf_{\mathcal F_T(\mathbf p,\tau)}
H(K|E)_\omega .
\end{equation}

\begin{theorem}[SDI-QKD security under an operational source assumption]
\label{thm:operational-key-rate}
In the asymptotic i.i.d. regime, every observed behavior \(\mathbf p\)
and every emitted ensemble satisfying \(T_{\mathcal Q}(\mathcal E)\leq\tau\)
obey
\begin{equation}
\label{eq:operational-key-rate}
r
\geq
\mathsf H_T(\mathbf p,\tau)
-
h_2(Q_Z).
\end{equation}
\end{theorem}

\begin{proof}
By Proposition~\ref{prop:channel-lift-source-assumptions}, every
physical Bob--Eve realization compatible with the source assumption and
the observed behavior belongs to \(\mathcal F_T(\mathbf p,\tau)\). Hence
\(H(K|E)_\omega\geq\mathsf H_T(\mathbf p,\tau)\) for every such
realization. Combining this with Eq.~\eqref{eq:dw-key-rate} gives
Eq.~\eqref{eq:operational-key-rate}.
\end{proof}

The theorem reduces security to a dimension-independent entropy
minimization over \(\mathcal F_T(\mathbf p,\tau)\). The next two
subsections give two certificates for this quantity.

\subsection{Min-entropy certificate}
\label{subsec:minentropy-certificate}

The min-entropy certificate bounds Eve's optimal probability of guessing
the retained key bit over the same feasible set
\(\mathcal F_T(\mathbf p,\tau)\). For the cq state in
Eq.~\eqref{eq:pm-key-reduction},
\begin{equation}
\label{eq:minentropy-definition}
H_{\min}(K|E)_\omega
=
-\log_2 p_{\mathrm{guess}}(K|E)_\omega .
\end{equation}
The corresponding worst-case guessing probability is
\begin{equation}
\label{eq:key-guessing-certificate}
\begin{aligned}
p_{\mathrm{guess},T}^{\star}(\mathbf p,\tau)
:=
\sup_{\substack{
\mathcal F_T(\mathbf p,\tau)\\
0\preceq E_0\preceq I_E
}}
\frac{1}{2}
\Bigl(
&
\operatorname{Tr}[\rho_0^E E_0]
\\[-1mm]
&+
\operatorname{Tr}[\rho_1^E(I_E-E_0)]
\Bigr),
\end{aligned}
\end{equation}
where \(\rho_k^E\) is the marginal defined in
Eq.~\eqref{eq:eve-key-branch-states} for the feasible Bob--Eve ensemble
being optimized.

A finite-level noncommutative moment relaxation gives a validated upper
bound
\begin{equation}
\label{eq:pguess-upper-bound}
p_{\mathrm{guess},T}^{\star}(\mathbf p,\tau)
\leq
\overline p_{\mathrm{guess},T}(\mathbf p,\tau).
\end{equation}
Since \(H(K|E)_\omega\geq H_{\min}(K|E)_\omega\), the certified
direct-reconciliation rate is
\begin{equation}
\label{eq:minentropy-rate}
r
\geq
r_T^{\min}(\mathbf p,\tau)
:=
-\log_2
\overline p_{\mathrm{guess},T}(\mathbf p,\tau)
-
h_2(Q_Z).
\end{equation}

This certificate has a direct operational meaning: it bounds Eve's
single-bit guessing probability. The moment relaxation and validation
conditions are given in Appendix~\ref{app:minentropy-sdp}.

\subsection{PM-BFF conditional-entropy certificate}
\label{subsec:pmbff-certificate}

The PM-BFF certificate lower-bounds the worst-case conditional entropy
\(\mathsf H_T(\mathbf p,\tau)\) in Eq.~\eqref{eq:worst-case-privacy}
without passing through Eve's guessing probability. We use a
prepare-and-measure version of the Brown--Fawzi--Fawzi variational
method for conditional entropy
\cite{BrownFawziFawzi2021ConditionalEntropies,BrownFawziFawzi2024ConditionalVonNeumannEntropy}. All entropies below are
in bits. For the retained-key cq state in
Eq.~\eqref{eq:pm-key-reduction},
\begin{equation}
\label{eq:bff-relative-entropy}
H(K|E)_\omega
=
-
D_{\rm rel}
\!\left(
\omega_{KE}
\,\middle\|\,
I_K\otimes\omega_E
\right),
\qquad
\omega_E
=
\frac{\rho_0^E+\rho_1^E}{2}.
\end{equation}
Here \(D_{\rm rel}\) is the Umegaki relative entropy. Write
\begin{equation}
\label{eq:bff-omega}
\Omega^E
:=
\rho_0^E+\rho_1^E
=
2\omega_E .
\end{equation}

Since \(K\) is classical, the BFF auxiliary operator on \(K\otimes E\)
may be chosen block diagonal in the key register. At a node
\(t\in(0,1)\), write it as
\[
Z_t
=
\sum_{k=0}^{1}
\lvert k\rangle\!\langle k\rvert_K
\otimes
Z_{k,t}.
\]
Substituting this form into the BFF node functional for
\(\omega_{KE}\) leaves one Eve-side auxiliary operator \(Z_{k,t}\) for
each key value. With \(\Omega^E=\rho_0^E+\rho_1^E\), the node
functional is
\begin{equation}
\label{eq:bff-node-functional}
\begin{aligned}
\beta_t(\rho_0^E,\rho_1^E)
:=
\inf_{Z_{0,t},Z_{1,t}}
\frac{1}{2}
\sum_{k=0}^{1}
\Bigl\{
&
\operatorname{Tr}\!\left[
\rho_k^E
\left(
Z_{k,t}+Z_{k,t}^{\dagger}
\right)
\right]
\\[-1mm]
&+
(1-t)
\operatorname{Tr}\!\left[
\rho_k^E Z_{k,t}^{\dagger}Z_{k,t}
\right]
\\[-1mm]
&+
t
\operatorname{Tr}\!\left[
\Omega^E Z_{k,t}Z_{k,t}^{\dagger}
\right]
\Bigr\}.
\end{aligned}
\end{equation}
The first two terms are branchwise. The last term contains the common
Eve marginal and is the only coupling between the two key branches.

For the finite node set \(\{t_i\}_{i=1}^{m}\subset(0,1)\), let
\(\alpha_i=w_i/(t_i\ln2)\) and define
\begin{equation}
\label{eq:bff-node-worst-case}
\beta_{T,i}^{\star}(\mathbf p,\tau)
:=
\inf_{\mathcal F_T(\mathbf p,\tau)}
\beta_{t_i}(\rho_0^E,\rho_1^E),
\end{equation}
where the states \(\rho_k^E\) are the retained-key marginals induced by
the feasible Bob--Eve ensemble.

\begin{proposition}[PM-BFF conditional-entropy certificate]
\label{prop:pmbff-entropy-certificate}
For every \(T\in\{D,\Pi,D\oplus A,\Pi\oplus A\}\),
\begin{equation}
\label{eq:bff-entropy-certificate}
\mathsf H_{T,m}^{\rm BFF}(\mathbf p,\tau)
:=
\sum_{i=1}^{m}
\alpha_i
\left[
1+
\beta_{T,i}^{\star}(\mathbf p,\tau)
\right]
\leq
\mathsf H_T(\mathbf p,\tau).
\end{equation}
Consequently,
\begin{equation}
\label{eq:bff-key-rate}
r
\geq
r_T^{\rm BFF}(\mathbf p,\tau)
:=
\mathsf H_{T,m}^{\rm BFF}(\mathbf p,\tau)
-
h_2(Q_Z).
\end{equation}
\end{proposition}

\begin{proof}
For each fixed feasible Bob--Eve realization, the
Brown--Fawzi--Fawzi variational formula and the chosen Gauss--Radau
finite-node approximation give
\[
H(K|E)_\omega
\geq
\sum_{i=1}^{m}
\alpha_i
\left[
1+
\beta_{t_i}(\rho_0^E,\rho_1^E)
\right].
\]
Taking the infimum over \(\mathcal F_T(\mathbf p,\tau)\) and using
\(\alpha_i>0\) gives Eq.~\eqref{eq:bff-entropy-certificate}. The rate
bound follows from Theorem~\ref{thm:operational-key-rate}.
\end{proof}

The node minimizers satisfy an Eve-side Sylvester equation. Varying
Eq.~\eqref{eq:bff-node-functional} with respect to
\(Z_{k,t}^{\dagger}\) gives
\begin{equation}
\label{eq:bff-sylvester-equation}
\rho_k^E
+
(1-t)Z_{k,t}\rho_k^E
+
t\Omega^E Z_{k,t}
=
0 .
\end{equation}
For the regularized interior-node problem used in the finite relaxation,
the same equation gives the dimension-independent auxiliary bound
\begin{equation}
\label{eq:bff-sylvester-bound}
\|Z_{k,t}\|
\leq
\frac{1}{2\sqrt{t(1-t)}} .
\end{equation}
In the noncommutative moment relaxation,
Eq.~\eqref{eq:bff-sylvester-equation} becomes a set of linear moment
constraints, while Eq.~\eqref{eq:bff-sylvester-bound} gives the
corresponding norm localizers. These constraints tighten the relaxation
without imposing any dimension bound on Eve. The construction and
finite-level validation are given in Appendix~\ref{app:pm-bff}.

For fixed \(T\), \(\mathbf p\), and \(\tau\), the min-entropy and
PM-BFF certificates are evaluated over the same feasible set
\(\mathcal F_T(\mathbf p,\tau)\). Their difference is the entropy
functional being certified: the min-entropy program bounds Eve's
single-bit guessing probability, while PM-BFF lower-bounds the
conditional von Neumann entropy entering the Devetak--Winter rate.
Changing \(D\) to \(D\oplus A\), or \(\Pi\) to \(\Pi\oplus A\), changes
the feasible set itself. The next section evaluates these certificates
for the source regimes identified above.

\section{Finite-level key-rate bounds}
\label{sec:keyrate-products}

The previous section leaves a concrete object to certify:
\(\mathsf H_T(\mathbf p,\tau)\), the worst-case conditional entropy over
the lifted Bob--Eve feasible set. The finite-level calculations below
separate three ingredients: the observed behavior, the operational
source assumption, and the entropy certificate. The behavior is always
the complete three-setting table \(\mathbf p\) of
Eq.~\eqref{eq:security-behavior}, the source assumption is always a
scalar operational bound \(T_{\mathcal Q}\leq\tau\), and Bob's
measurements remain uncharacterized.

The min-entropy relaxation certifies the optimized retained-key guessing
probability \(p_{\mathrm{guess},T}^{\star}(\mathbf p,\tau)\). The PM-BFF
relaxation certifies the conditional von Neumann entropy entering the
Devetak--Winter rate. All PM-BFF curves use one joint Bob--Eve
relaxation, with the selected source task entering through its cover
constraints and trace bound. All min-entropy curves use one degree-three
guessing-probability relaxation. Every finite-level PM-BFF or
min-entropy point plotted below is an accepted certificate satisfying the
verification criteria in Appendix~\ref{app:keyrate-numerics}.

We first freeze the operational source assumption at \(D\leq1/2\), the
operational consequence of a qubit source. This is weaker than imposing
dimension two. Any improvement in this benchmark therefore comes from
the protocol or from the entropy certificate, not from stronger source
information. We then replace this fixed implication by the actual
operational task values of the emitted ensemble, introduce receiver-side
depolarization of the reference decoders used to generate the observed
behavior, and finally test direct-sum leakage.
Table~\ref{tab:critical-visibility-summary} collects the critical source
visibilities used in the comparisons below.

\begin{table}[t]
\footnotesize
\begin{ruledtabular}
\begin{tabular}{l c c}
source assumption and protocol & min-entropy & PM-BFF\\
\hline
$D\leq1/2$, RAC key, linear \cite{PawlowskiBrunner2011SDIQKD} & $0.9659$ & ---\\
$D\leq1/2$, RAC key, squared \cite{RajPrasadChaturvediEtAl2026WaveParticleSDI} & $0.9428$ & ---\\
$D\leq1/2$, retained key & $0.8895$ & $0.8741$\\
$\Pi=1/2$ & $0.8895$ & $0.8546$\\
$D^{(\nu)}$ & $0.8162$ & $0.7202$\\
$(D\oplus A)^{(\nu)}$ & $0.0372$ & $0.0106$\\
$(\Pi\oplus A)^{(\nu)}$ & $0.0526$ & $0.0112$\\
\end{tabular}
\end{ruledtabular}
\caption{Critical source visibilities for positive certified key under
preparation depolarization. The first two rows are the earlier RAC-key
protocols, where the RAC output itself is the raw key. All remaining rows
use the three-setting retained-key protocol and the complete observed
three-setting behavior. Each finite-level value is the linear root estimate
bracketed by accepted samples of opposite sign, as described in
Appendix~\ref{app:keyrate-numerics}.}
\label{tab:critical-visibility-summary}
\end{table}

\subsection{Dimension-implied distinguishability benchmark}
\label{subsec:dimension-implied-D}

By Eq.~\eqref{eq:qubit-implies-D-main}, a qubit source implies
\(D\leq1/2\), but the converse is false. We impose only this operational
consequence. The feasible set therefore contains every emitted ensemble,
in arbitrary Hilbert-space dimension, satisfying the same
distinguishability bound.

The visibility \(\nu\) in this benchmark parametrizes depolarization of
the observed behavior, not the source constraint. For every point on the
orange and blue curves, the operational source assumption remains the fixed bound
\(D\leq1/2\). The comparison therefore isolates two effects before any
additional source information is used. Earlier qubit-source SDI
protocols use the RAC output itself as the raw key, so their
reconciliation cost is fixed by Bob's RAC-output guessing probability.
Here the RAC settings test the source, while the additional setting
\(y=2\) generates the retained key on the branch defined in
Eq.~\eqref{eq:key-branch}; for ideal BB84 this key branch has \(Q_Z=0\). The second effect is the entropy certificate: the
min-entropy curve uses the retained-key optimizer
\(p_{\mathrm{guess},T}^{\star}(\mathbf p,\tau)\), while PM-BFF certifies
the conditional von Neumann entropy directly.

For the same depolarized RAC statistics, Bob's RAC-output guessing
probability is
\(p_{\mathrm{guess}}^{B,\mathrm{RAC}}(\nu)=1/2+\nu/(2\sqrt2)\). The
linear RAC-key proof gives the Eve-side bound
\(p_{\mathrm{guess}}^{E,\mathrm{RAC},\mathrm{lin}}(\nu)
=(3+\sqrt3)/4-\nu/(2\sqrt2)\), and the squared refinement gives
\(p_{\mathrm{guess}}^{E,\mathrm{RAC},\mathrm{sq}}(\nu)
=(3+\sqrt{1-\nu^2})/4\). The superscript \(\mathrm{RAC}\) records the
random variable being guessed: the RAC output used as the raw key in the
earlier protocols. Applying the Csiszár--Körner expression to these two
Eve-side bounds gives the gray and green curves in
Fig.~\ref{fig:visibility-Dlehalf-comparison}. At \(\nu=1\), the two
rates are \(0.0581\) and \(0.2104\), respectively; their penalty term is
the RAC-output reconciliation cost of the older protocols.

\begin{figure}[t]
\centering
\includegraphics[width=\columnwidth]{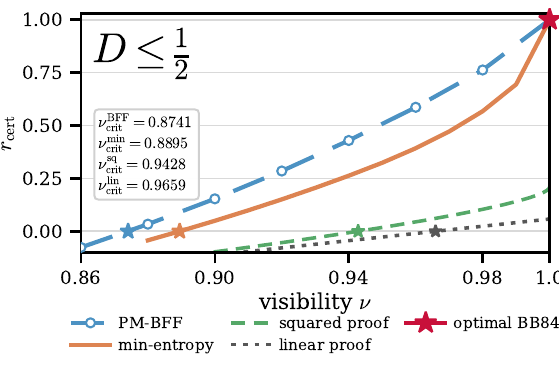}
\caption{Fixed distinguishability benchmark \(D\leq1/2\). Gray and green curves
are the earlier RAC-key Csisz\'ar--K\"orner rates obtained from the
RAC-output guessing bounds stated in the text. Orange and blue curves use
the present retained-key protocol under the same operational source
assumption: orange is the min-entropy certificate and blue is PM-BFF. The
critical visibilities are \(\nu_{\rm crit}^{\rm lin}=0.9659\),
\(\nu_{\rm crit}^{\rm sq}=0.9428\), \(\nu_{\rm crit}^{\min}=0.8895\),
and \(\nu_{\rm crit}^{\rm BFF}=0.8741\); the red star marks ideal BB84.}
\label{fig:visibility-Dlehalf-comparison}
\end{figure}

Figure~\ref{fig:visibility-Dlehalf-comparison} shows the ordering under
the sole operational source assumption \(D\leq1/2\). The retained-key
protocol moves the min-entropy threshold below the two earlier RAC-key
curves. PM-BFF lowers it further, from \(0.8895\) to \(0.8741\). The next
subsection keeps the security model fixed and changes only the
operational source assumption: the fixed implication \(D\leq1/2\) is
replaced by the actual operational task values of the emitted ensemble.
\subsection{Source depolarization under operational assumptions}
\label{subsec:source-depolarization}

The fixed-bound benchmark used only the one-way implication
\(D\leq1/2\). Here the security optimization uses the
preparation-depolarized BB84 family \(\rho_x^{(\nu)}\) from
Eq.~\eqref{eq:bb84-source-families}. Its RAC score and operational
task values are those of Eq.~\eqref{eq:bb84-visibility-task-values}. The
three-setting behavior \(\mathbf p^{(\nu)}\) is obtained by adding the
key setting of Eq.~\eqref{eq:key-branch} to the same four-preparation
ensemble, with the retained-branch error evaluated by
Eq.~\eqref{eq:key-qber}.

For each
\(T\in\{D,\Pi,D\oplus A,\Pi\oplus A\}\), the corresponding run imposes
the single scalar bound \(T\leq T^{(\nu)}\), where \(T^{(\nu)}\) is the
operational task value listed in
Eq.~\eqref{eq:bb84-visibility-task-values}. Thus the distinguishability
task value varies with the contraction of the four preparations, the
parity value stays fixed because the even and odd parity mixtures
coincide, and the composite assumptions inherit their visibility
dependence from the state-exclusion value \(A^{(\nu)}\).

\begin{figure*}[t]
\centering
\includegraphics[width=\textwidth]{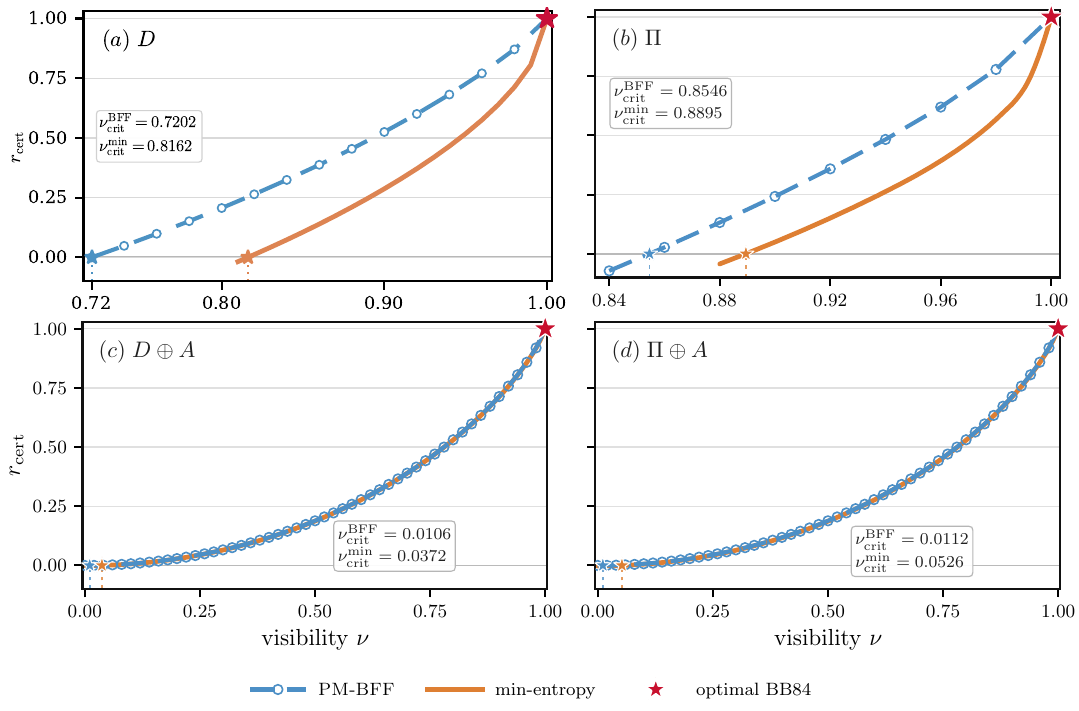}
\caption{Source-depolarized BB84 key-rate lower bounds under four operational
source assumptions. Dashed blue curves with open circles are PM-BFF
bounds; solid orange curves are min-entropy bounds. The blue and orange
axis stars are linear root estimates from the nearest validated
samples of opposite sign and are not additional SDP evaluations. Lines
connect accepted certificates; the red star marks ideal BB84.}
\label{fig:keyrate-visibility}
\end{figure*}

The top-left panel of Fig.~\ref{fig:keyrate-visibility} gives the direct
comparison with the fixed-bound benchmark. For \(D^{(\nu)}\), the PM-BFF
critical visibility drops from \(0.8741\) to \(0.7202\), while the
min-entropy threshold is \(0.8162\). The security model has not changed:
the improvement comes from replacing the dimension-implied ceiling
\(D\leq1/2\) by the visibility-dependent bound \(D\leq D^{(\nu)}\).

The parity panel uses a different coarse-graining of the same emitted
ensemble. Since the family is parity-oblivious at the source-task level,
\(\Pi^{(\nu)}=1/2\) for all \(\nu\). Parity discrimination therefore does not
track the contraction of the four individual preparations. The
thresholds, \(0.8546\) for PM-BFF and \(0.8895\) for min-entropy, differ
from those for \(D^{(\nu)}\) because four-state and parity discrimination constrain different coarse-grainings of the source ensemble.

The lower panels add state exclusion. The exclusion task bounds how well
a measurement can rule out a value of $x$, in addition to how well
it can identify one. Since \(A^{(\nu)}\) is visibility-dependent in
Eq.~\eqref{eq:bb84-visibility-task-values}, both
\((D\oplus A)^{(\nu)}\) and \((\Pi\oplus A)^{(\nu)}\) vary with
visibility; in particular, \(\Pi\oplus A\) changes although
\(\Pi^{(\nu)}\) itself is constant. The resulting critical visibilities
are near zero: \(0.0106\) and \(0.0112\) for PM-BFF, with min-entropy
thresholds \(0.0372\) and \(0.0526\).

All four panels use the same security model and the same three-setting
behavior \(\mathbf p^{(\nu)}\). Only the operational source assumption
changes. The figure separates operational task values from
entropy certification: \(D^{(\nu)}\) improves over the fixed
dimension-implied bound, state exclusion gives the largest robustness
gain, and PM-BFF enlarges the positive-rate region for every source
assumption shown.

The source-depolarized family also has a direct optical reading in the
single-photon limit. Consider a phase-randomized source whose common
vacuum component has weight \(q\), whose non-vacuum component realizes
the BB84 qubit ensemble, and whose no-click outcomes are assigned a
uniform random bit. The vacuum contributes the measurement-independent
values \((1/4,1/2,3/4)\) to four-state discrimination, parity discrimination, and
state exclusion, so the operational task values and honest behavior
coincide with those of \(\rho_x^{(\nu)}\) at \(\nu=1-q\). At the saturated
Poisson value \(q=e^{-\mu}=1/2\), equivalently \(\mu=\ln2\), this gives
\(\nu=1/2\): the exclusion-assisted PM-BFF thresholds \(0.0106\) and
\(0.0112\) are far below this value, while the identification-only
thresholds \(0.7202\) and \(0.8546\) are not. For multiphoton weak
coherent implementations, the non-vacuum sector is no longer a qubit
BB84 ensemble; the universal common-component bound of
Eq.~\eqref{eq:common-component-composite-main} is then the applicable
source statement.

\subsection{Receiver-side depolarization}
\label{subsec:receiver-side-depolarization}

We next keep Alice's emitted ensemble fixed as \(\rho_x^{(\nu)}\) from
Eq.~\eqref{eq:bb84-source-families}. Hence the operational source
task values remain those of Eq.~\eqref{eq:bb84-visibility-task-values}.
Only the reference decoders used to generate the observed behavior are
depolarized, as in Eq.~\eqref{eq:noisy-bb84-decoders}; the resulting RAC
score is Eq.~\eqref{eq:joint-source-receiver-rac-score}. The retained
branch is still the one in Eq.~\eqref{eq:key-branch}, and its error is
evaluated by Eq.~\eqref{eq:key-qber}. Thus \(\nu_{\rm M}\) parametrizes
the observed table, not a trusted property of Bob's device in the
security optimization.

For each certificate and operational source assumption, we extract the
critical source visibility at fixed receiver-side visibility
\(\nu_{\rm M}\). Each point is obtained from the nearest accepted samples
with opposite key-rate sign and is not an additional SDP evaluation.

\begin{figure}[t]
\centering
\includegraphics[width=\columnwidth]{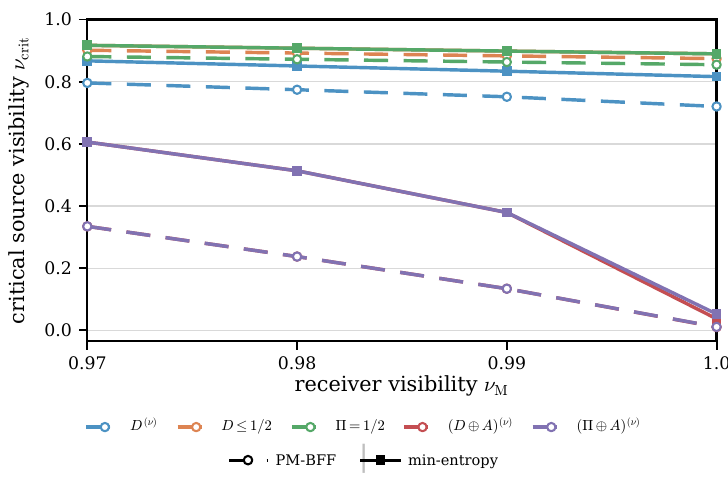}
\caption{Receiver-side depolarization. For each receiver-side visibility
\(\nu_{\rm M}\), the vertical axis gives the critical source visibility at
which the certified key rate becomes positive. Colors identify the
operational source assumption; dashed open-circle curves are PM-BFF
bounds and solid square-marker curves are min-entropy bounds. The
\(\nu_{\rm M}=1\) endpoints reproduce the source-only thresholds in
Fig.~\ref{fig:keyrate-visibility}.}
\label{fig:joint-source-receiver}
\end{figure}

Figure~\ref{fig:joint-source-receiver} shows that the ordering found in
the source-only curves persists under receiver-side depolarization. The
fixed assumptions \(D\leq1/2\) and \(\Pi=1/2\) remain close to unit
source visibility. The exact distinguishability value
\(D^{(\nu)}\) gives lower thresholds because the source assumption
follows the emitted ensemble instead of using only the dimension-implied
ceiling. The exclusion-assisted assumptions remain the most robust
throughout the displayed range. At \(\nu_{\rm M}=0.99\), their PM-BFF
critical source visibility is \(0.1341\), while the corresponding
min-entropy values are \(0.3792\) and \(0.3793\).

Receiver-side depolarization changes only the observed behavior. The
operational source assumption remains a constraint on Alice's emitted
ensemble, and Bob's device is still uncharacterized in the security
optimization. Thus this benchmark separates degradation of the observed
evidence from degradation of the source-task values.
\subsection{Direct-sum leakage}
\label{subsec:direct-sum-leakage}

The last benchmark keeps the retained key branch perfect while exposing
the value of $x$ through the emitted ensemble. The emitted ensemble
is the direct-sum family \(\rho_x^{(\ell)}\) from
Eq.~\eqref{eq:bb84-source-families}, in which an orthogonal sector
carries the value of $x$ with weight \(\ell\). Its RAC score and
operational task values are those of Eq.~\eqref{eq:bb84-leakage-task-values}.
The same retained branch of Eq.~\eqref{eq:key-branch} gives
\(Q_Z^{(\ell)}=0\) through Eq.~\eqref{eq:key-qber}.

The four operational source assumptions are optimized independently. The
coincidence of their PM-BFF bounds, and separately of their min-entropy
bounds, is therefore a numerical outcome of four distinct feasible-set
optimizations; no equality between the source constraints is imposed.

\begin{figure}[t]
\centering
\includegraphics[width=\columnwidth]{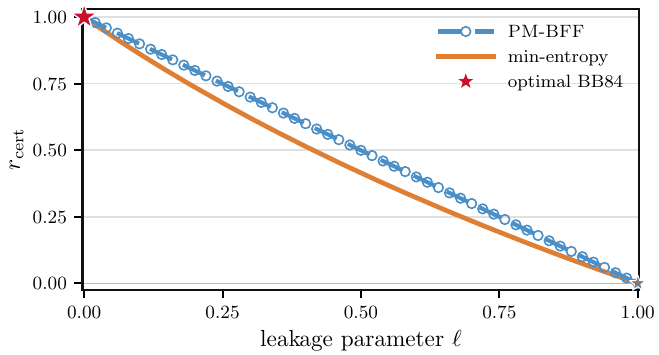}
\caption{Direct-sum leakage. The four operational source assumptions are
optimized independently; their PM-BFF bounds coincide with one another,
and their min-entropy bounds coincide with one another, within the
verification tolerance of Appendix~\ref{app:keyrate-numerics}. The
retained key branch has \(Q_Z^{(\ell)}=0\), while the orthogonal sector
reveals an increasing fraction of the value of $x$. At \(\ell=1\),
Eve can copy $x$ in full without changing Bob's statistics, and the
exact worst-case rate is zero.}
\label{fig:keyrate-leakage}
\end{figure}

Figure~\ref{fig:keyrate-leakage} shows that zero key-basis error is not
sufficient when the emitted ensemble carries explicit information about $x$.
For every sampled \(\ell<1\), the finite-level certificates remain
positive. At \(\ell=1\), the orthogonal sector contains the complete
value of $x$. Eve can copy that sector, infer the retained key, and
leave Bob's statistics unchanged; the exact worst-case rate is therefore
zero. The finite-level bounds reproduce this endpoint within the
verification tolerance.

The benchmark sequence separates three effects. Source depolarization
changes the operational task values of Alice's emitted ensemble.
Receiver-side depolarization weakens the observed behavior while leaving
those task values fixed. Direct-sum leakage keeps
\(Q_Z=0\) but increases source-side information about $x$. Across all three
families, the operational source assumption determines the admissible
Bob--Eve set, and the certificate determines which entropy quantity is
lower-bounded on that set. The min-entropy relaxation bounds
\(p_{\mathrm{guess},T}^{\star}(\mathbf p,\tau)\); PM-BFF lower-bounds
the conditional von Neumann entropy entering the Devetak--Winter rate.
\section{Discussion}
\label{sec:discussion}

The trusted object in this work is Alice's emitted ensemble. The source
assumption is not a statement about the size of the Hilbert space, nor a
specification of the devices that prepare or measure the states. It is a
scalar operational bound on what the four emitted preparations make
possible. The task values \(D\), \(\Pi\), and \(A\) ask three different
questions of the same ensemble: how well the value of $x$ can be identified,
how well its parity can be identified, and how well one state can be
excluded. Proposition~\ref{prop:exclusion-not-fixed} shows that the
state-exclusion task is independent of the identification task values:
two ensembles can have the same four-state and parity optima and still
differ in what a measurement can rule out. The normalized composites
\(D\oplus A\) and \(\Pi\oplus A\) are therefore single operational source
assumptions, not separate component constraints.

This viewpoint fixes the role of the parity-oblivious endpoint. Under
the standard relabelling \(x=(a,a\oplus s)\), the parity
\(x_0\oplus x_1\) is Alice's Bell setting \(s\). The condition
\(\Pi=1/2\) is exactly the statement that the two effective states
corresponding to \(s=0\) and \(s=1\) are identical. Equivalently,
\(\rho_{00}+\rho_{11}=\rho_{01}+\rho_{10}\), which is the
no-signalling identity for the remotely prepared ensemble. Conversely,
this identity makes the two pairs \(\{\rho_{00},\rho_{11}\}\) and
\(\{\rho_{01},\rho_{10}\}\) two decompositions of the same average state;
a purification of that average state realizes them by two measurements on
Alice's side. Thus the \(\Pi=1/2\) restriction is the prepare-and-measure
representation of the corresponding no-signalling Bell slice
\cite{WrightFarkas2023BellContextualityMap}. At this endpoint, the RAC
score under the same relabelling is the CHSH winning probability, as
shown in Appendix~\ref{app:operational-origins}. In the same relabelling,
\(\Pi>1/2\) quantifies how well the emitted ensemble reveals Alice's Bell
setting. A bound on \(\Pi\) is therefore a bound on signalling from
Alice's setting to the receiver, expressed as an operational
one-shot discrimination probability.

This also clarifies the relation with receiver-device-independent QKD.
RDI protocols based on overlap assumptions treat Bob as a black box and
replace a dimension bound by a bound on the distinguishability of
Alice's preparations
\cite{IoannouEtAl2022RDIQKD,IoannouSekatskiAbbottRossetBancalBrunner2022RDIProtocols}.
For mixed preparations, the overlap formulation must refer to
purifications satisfying the assumed Gram constraints, with the purifying
system assigned to Alice's laboratory; possible side-channel information
is then absorbed by modifying the overlaps. The present formulation
removes that purification choice from the assumption. The bound is placed
directly on the full emitted density operators, including every
$x$-dependent flag, side system, or optical mode that leaves Alice. If
such degrees of freedom reveal the value of $x$, the optimized source task sees
that revelation. The data-processing lift in
Proposition~\ref{prop:channel-lift-source-assumptions} then carries the
same scalar bound to every Bob--Eve extension. In this precise sense, the
mixed-state issue is handled natively: the proof does not require a pure
source model, a selected purification, or a dimension bound.

The classical-frontier analysis gives these source-task values an
operational baseline. Theorem~\ref{thm:classical-frontiers} describes the
best classical, or commuting, RAC behavior compatible with each task
bound. A positive frontier deviation certifies that, in every compatible
quantum realization, Alice's preparations cannot all commute and Bob's
RAC measurements cannot be jointly measurable. The security proof then
uses the same task value again, now lifted to the Bob--Eve feasible set.
Thus the source tasks are not auxiliary diagnostics: they are the common
operational language connecting prepare-and-measure nonclassicality, physical
source restrictions, and adversarial entropy certification.

The numerical section separates three effects that are often conflated.
First, the three-setting protocol removes the intrinsic RAC-key
reconciliation penalty of earlier two-setting SDI-QKD constructions:
\(y=0,1\) test the RAC, while \(y=2\) generates the retained key on the
\(00/11\) branch. Second, PM-BFF strengthens the entropy certificate on a
fixed Bob--Eve feasible set by lower-bounding the conditional von Neumann
entropy directly, whereas the min-entropy program bounds Eve's optimized
single-bit guess. Third, changing from \(D\) or \(\Pi\) to
\(D\oplus A\) or \(\Pi\oplus A\) changes the feasible set itself. The
large robustness gain comes primarily from this third effect.

The cryptographic role of exclusion is therefore structural. Four-state
and parity discrimination constrain what an adversary can identify. State
exclusion constrains what an adversary can rule out. These capabilities
are related, but one does not determine the other. An adversary who
cannot identify the emitted state may still exclude enough alternatives
to reduce the uncertainty relevant for the retained key. Conversely,
bounding exclusion removes Bob--Eve extensions that remain admissible
under identification-only assumptions. This is why the
exclusion-assisted source assumptions give a much larger improvement
than the change from one entropy certificate to another.

This use of exclusion is consistent with wider roles of exclusion tasks,
but the security proof does not rely on those frameworks. Exclusion tasks
operationally characterize weight-based resource quantifiers in convex
resource theories, antidistinguishability gives one-way communication
separations, and state discrimination admits contextual advantages
\cite{DucuaraSkrzypczyk2020WeightExclusion,UolaBullockKraftPellonpaaBrunner2020AllResourcesExclusion,HavlicekBarrett2020AntidistinguishabilityCC,SchmidSpekkens2018ContextualDiscrimination}.
Here these results serve only to locate the operational assumption: the proof
uses the scalar task bound on the emitted ensemble.

The benchmarks make the separation explicit. Under preparation
depolarization, \(D\) and \(\Pi\) give positive key only above sizable
source visibilities, while \(D\oplus A\) and \(\Pi\oplus A\) remain
positive down to nearly vanishing source visibility. Receiver-side
depolarization weakens only the complete observed table, not Alice's
source-task value, and the same ordering persists. Direct-sum leakage gives the complementary limit: the retained key branch remains
error free, but the emitted ensemble increasingly reveals $x$. The
finite-level bounds stay positive at every sampled point with incomplete leakage and vanish only at complete leakage. Zero key-basis error
is therefore not the security criterion; the relevant object is the
emitted ensemble constrained together with the complete three-setting
behavior.

Physical source assumptions provide the experimentally meaningful entry
point to the framework. They are not imposed in the security proof as
microscopic constraints; each is used to derive an operational bound on
the emitted ensemble. Finite-dimensional support, overlap or fidelity
information, and mean-energy or photon-number restrictions can all bound
four-state discrimination. Exact parity obliviousness gives the
parity-discrimination endpoint \(\Pi=1/2\), equivalently the
no-signalling Bell slice discussed above. Trusted global phase
randomization together with a common vacuum component gives the
exclusion-assisted optical bounds by substitution into
Eq.~\eqref{eq:common-component-composite-main}, with the detailed
derivation in Appendix~\ref{app:operational-origins}. These are
sufficient physical origins of the operational assumptions, not
additional device assumptions inside the adversarial optimization. Once
the scalar operational bound is fixed, the security proof uses only that
bound, the complete three-setting behavior, and the data-processing lift
to unrestricted Bob--Eve extensions. Bob's device remains
uncharacterized, and Eve's Hilbert space remains unrestricted.

The rates reported here are asymptotic i.i.d. rates against collective
attacks, per retained key round and with ideal one-way reconciliation.
The accepted PM-BFF dual certificates nevertheless already have the form
required for a finite-key reduction: they give affine one-round lower
bounds on the conditional entropy as functions of the complete observed
behavior and the source bound. Appendix~\ref{app:finite-key-geat} proves
that such certificates are valid min-tradeoff inputs for generalized
entropy accumulation
\cite{MetgerFawziSutterRenner2024GEAT,MetgerRenner2023GEATQKDSecurity}.
Explicit finite-\(n\) rates for the present protocol still require a
chosen parameter-estimation test, error-correction leakage, smoothing
parameters, and privacy-amplification constants. These finite-key
ingredients are separate from the conceptual point: the trusted part of
an SDI protocol should be stated in the operational language of the
emitted ensemble.

A natural continuation is the self-testing of ensembles. Operational
source assumptions should not remain only hypotheses inserted into a
security proof; they should become certifiable features of the emitted
ensemble. In Bell scenarios, self-testing identifies states and
measurements from correlations up to the relevant equivalences. In
prepare-and-measure scenarios, the object to be certified is the
ensemble geometry relevant to the optimized tasks. The direction opened
by the present work is therefore to combine operational
nonclassicality of ensembles, self-testing of ensembles, and entropy
certification in a single SDI framework. Robust SDI security is governed by what Eve can identify and by what
Eve can exclude.

\section*{Data availability}
The numerical data, figure data, source code, and saved certificate summaries used to generate every reported curve are included in the reproducibility archive accompanying this manuscript.
\section*{Acknowledgements}
A.C. acknowledges support from the KLAR Grant No. BNI/PST/2023/1/00013/U/00001, funded by NAWA. DS acknowledges support from STARS (STARS/STARS-2/2023-0809), Govt. of India. G.V. acknowledges support from the Deutsche Forschungsgemeinschaft
(DFG, German Research Foundation, project number 563437167), the
Sino-German Center for Research Promotion (Project M-0294),
and the German Federal Ministry of Research, Technology and Space
(Project QuKuK, Grant No. 16KIS1618K and Project BeRyQC, Grant No.
13N17292). E.P. acknowledges support by the project skQCI (Slovak Quantum Communication Infrastructure), co-funded by the European Union under the Digital Europe Programme, grant agreement No. 101091548, as part of the EuroQCI initiative.
T.P acknowledges support from the project 'International Center for Theory of Quantum Technologies 2.0: R\&D Industrial and Experimental Phase' project (contract number FENG.02.01-IP.05-0006/23). The project is implemented as part of the International Research Agendas Programme of the Foundation for Polish Science, co-financed by the European Funds for a Smart Economy 2021-
2027 (FENG), Priority FENG.02 Innovation-friendly environment, Measure FENG.02.01.
\clearpage
\appendix

\section{Classical posterior geometry and tightness}
\label{app:classical-frontiers}

This appendix proves Theorem~\ref{thm:classical-frontiers}. The proof uses only the posterior distribution induced by a classical message and places no bound on the finite message alphabet. The argument has three steps. First, the classical source is decomposed into one-message posteriors. Second, each posterior gives a finite optimization on the simplex over four values of $x$. Third, four elementary sources attain all exposed boundary segments.

\subsection{From messages to posteriors}

A classical source is specified by conditional probabilities \(p(\lambda|x)\), where \(x\in\mathcal X\) is Alice's string $x$ and \(\lambda\) is the message received by Bob. Since the value of $x$s are sampled uniformly, the probability of the message \(\lambda\) is
\begin{equation}
\label{eq:app-mu}
\mu_\lambda
=
\frac14\sum_{x\in\mathcal X}p(\lambda|x).
\end{equation}
Equation~\eqref{eq:app-mu} is just the law of total probability with prior \(p(x)=1/4\). It is the weight with which the message \(\lambda\) contributes to every observed or optimized classical quantity.

For every message with \(\mu_\lambda>0\), Bayes' rule gives the posterior distribution over values of \(x\),
\begin{equation}
\label{eq:app-posterior}
q_x^\lambda
=
p(x|\lambda)
=
\frac{p(\lambda|x)}{4\mu_\lambda}.
\end{equation}
The denominator is precisely the normalization from Eq.~\eqref{eq:app-mu}. Messages with \(\mu_\lambda=0\) are never observed and can be removed. Thus an arbitrary classical implementation is a convex mixture of one-message posteriors \(q^\lambda\).

For one such posterior we write
\begin{equation}
\label{eq:app-q}
q=(q_{00},q_{01},q_{10},q_{11}),
\qquad
q_x\ge0,
\qquad
\sum_{x\in\mathcal X}q_x=1.
\end{equation}
This equation fixes the simplex on which the rest of the proof takes place. Each entry is the conditional probability of one of Alice's four values of $x$ after the classical message is known.

The optimal one-message RAC score is obtained by answering each question with the more likely value of the requested bit:
\begin{equation}
\label{eq:app-rq}
\begin{aligned}
r(q)=\frac12\Big(&
\max\{q_{00}+q_{01},\,q_{10}+q_{11}\}
\\
&+
\max\{q_{00}+q_{10},\,q_{01}+q_{11}\}
\Big).
\end{aligned}
\end{equation}
The first maximum is the best success probability when Bob is asked for \(x_0\); the second is the best success probability when he is asked for \(x_1\). The prefactor \(1/2\) is the uniform choice of Bob's question.

The three one-message source task values are
\begin{equation}
\label{eq:app-dpa}
\begin{aligned}
d(q)&=\max_{x\in\mathcal X}q_x,\\
\pi(q)&=\max\{q_{00}+q_{11},\,q_{01}+q_{10}\},\\
a(q)&=1-\min_{x\in\mathcal X}q_x .
\end{aligned}
\end{equation}
The first line is optimal four-state discrimination. The second line is optimal parity discrimination, because the two sets \(\{00,11\}\) and \(\{01,10\}\) are the two parity classes. The third line is optimal exclusion: to rule out the true value of $x$, one names the least likely value of \(x\).

Since Bob knows \(\lambda\), the global operational quantities are affine averages of the one-message quantities:
\begin{equation}
\label{eq:app-affine}
\begin{aligned}
R&=\sum_\lambda\mu_\lambda r(q^\lambda),&
D&=\sum_\lambda\mu_\lambda d(q^\lambda),\\
\Pi&=\sum_\lambda\mu_\lambda \pi(q^\lambda),&
A&=\sum_\lambda\mu_\lambda a(q^\lambda).
\end{aligned}
\end{equation}
Equation~\eqref{eq:app-affine} is the reduction of the classical problem to a one-posterior problem. Once an inequality is proved for each posterior \(q\), the same inequality holds after averaging over \(\lambda\).

\subsection{A normal form for one posterior}

The RAC score can be written as a chance value plus the two optimal bit biases:
\begin{equation}
\label{eq:app-rq-absolute}
\begin{aligned}
r(q)=\frac12+\frac14\Big(&
\left|q_{00}+q_{01}-q_{10}-q_{11}\right|
\\
&+
\left|q_{00}+q_{10}-q_{01}-q_{11}\right|
\Big).
\end{aligned}
\end{equation}
The absolute values encode Bob's freedom to decide which bit value is more likely for each question. Flipping either bit value only relabels the posterior entries; it does not change \(d(q)\), \(a(q)\), or \(\pi(q)\), because parity classes are merely exchanged under a single bit flip.

We may therefore work in a sector where the optimal answer is \(0\) for both questions. Write
\begin{equation}
\label{eq:app-abgd}
\begin{aligned}
q_{00}&=\alpha, & q_{01}&=\beta,
& q_{10}&=\gamma, & q_{11}&=\delta,\\
\alpha&+\beta+\gamma+\delta=1.
\end{aligned}
\end{equation}
In this sector the probability weight on values of \(x\) with \(x_0=0\) is at least the weight on values of \(x\) with \(x_0=1\), and similarly for \(x_1\). Hence
\begin{equation}
\label{eq:app-sector}
\alpha+\beta\ge\gamma+\delta,
\qquad
\alpha+\gamma\ge\beta+\delta.
\end{equation}
Substituting these two sector conditions into Eq.~\eqref{eq:app-rq} gives
\begin{equation}
\label{eq:app-r-alpha-delta}
r(q)=\frac12+\frac{\alpha-\delta}{2}.
\end{equation}
Indeed, the two selected bit marginals are \(\alpha+\beta\) and \(\alpha+\gamma\), whose sum is \(1+\alpha-\delta\). Thus the RAC advantage over chance is exactly the separation between the posterior corner \(00\) and the opposite corner \(11\).

Adding the two sector inequalities in Eq.~\eqref{eq:app-sector} gives
\begin{equation}
\label{eq:app-alpha-ge-delta}
\alpha\ge\delta.
\end{equation}
Hence the favoured corner is never lighter than the opposite corner in the chosen sector.

\subsection{The four-state-discrimination frontier}

Let
\begin{equation}
\label{eq:app-d-def}
d=\max\{\alpha,\beta,\gamma,\delta\}.
\end{equation}
This is the one-message value of four-state discrimination. Since \(\alpha,\beta,\gamma\le d\), normalization implies
\begin{equation}
\label{eq:app-delta-lower-d}
\delta=1-\alpha-\beta-\gamma\ge1-3d.
\end{equation}
Since no posterior atom exceeds \(d\), normalization forces the fourth atom to satisfy \(\delta\geq1-3d\).

Using \(\alpha\le d\) and Eq.~\eqref{eq:app-delta-lower-d} in Eq.~\eqref{eq:app-r-alpha-delta} gives
\begin{equation}
\label{eq:app-r-le-2d}
r(q)
\le
\frac12+\frac{d-(1-3d)}2
=2d.
\end{equation}
This is the small-\(D\) branch. It is strongest when the posterior is still close to uniform, where the lower bound on \(\delta\) is active.

A second bound follows from \(\alpha\le d\) and \(\delta\ge0\):
\begin{equation}
\label{eq:app-r-le-half-d}
r(q)
\le
\frac12+\frac d2
=
\frac{1+d}{2}.
\end{equation}
This is the large-\(D\) branch. Averaging Eqs.~\eqref{eq:app-r-le-2d} and \eqref{eq:app-r-le-half-d} over \(\lambda\) gives \(R\le2D\) and \(R\le(1+D)/2\). Their lower envelope is \(F_D\), and the two lines cross at \(D=1/3\).

\subsection{The parity frontier}

Let
\begin{equation}
\label{eq:app-pi-def}
\pi=\max\{\alpha+\delta,\,\beta+\gamma\}.
\end{equation}
This is the one-message value of parity discrimination: the two entries in the maximum are the posterior weights of the even and odd parity classes. Since \(\pi\ge\alpha+\delta\), and \(\delta\ge0\),
\begin{equation}
\label{eq:app-alpha-delta-pi}
\alpha-\delta\le\alpha+\delta\le\pi.
\end{equation}
Substituting this into the normal form Eq.~\eqref{eq:app-r-alpha-delta} yields
\begin{equation}
\label{eq:app-r-pi}
r(q)\le\frac12+\frac\pi2=\frac{1+\pi}{2}.
\end{equation}
Averaging over messages gives \(R\le(1+\Pi)/2\), which is the parity frontier.

\subsection{The four-state-discrimination-with-exclusion frontier}

Let
\begin{equation}
\label{eq:app-m-a-def}
m=\min\{\alpha,\beta,\gamma,\delta\},
\qquad
a=1-m.
\end{equation}
The least posterior weight \(m\) is the value of $x$ one should exclude, so \(a\) is the one-message exclusion probability. Combining \(\alpha\le d\) with \(\delta\ge m\) gives
\begin{equation}
\label{eq:app-alpha-delta-da}
\alpha-\delta\le d-m.
\end{equation}
Using Eq.~\eqref{eq:app-r-alpha-delta}, this gives
\begin{equation}
\label{eq:app-r-da}
r(q)
\le
\frac12+\frac{d-m}{2}
=
\frac{d+a}{2}.
\end{equation}
Writing \(t=D\oplus A=(D+A)/2\), averaging over messages gives
\begin{equation}
\label{eq:app-r-doplusA}
R\le t= D\oplus A.
\end{equation}
The exclusion term enters because it controls how small the opposite corner \(\delta\) can be.

\subsection{The parity-with-exclusion frontier}

The parity-with-exclusion frontier is the lower envelope of two affine constraints. The first one-message constraint is
\begin{equation}
\label{eq:app-r-pia-half}
r(q)\le\frac{\pi+a}{2}.
\end{equation}
To prove it, Eq.~\eqref{eq:app-r-alpha-delta} shows that it is enough to bound \(\alpha-\delta\) by \(\pi-m\). Since \(\pi\ge\alpha+\delta\) and \(\delta\ge m\),
\begin{equation}
\label{eq:app-pi-minus-m}
\pi-m
\ge
\alpha+\delta-m
\ge
\alpha-\delta.
\end{equation}
This proves Eq.~\eqref{eq:app-r-pia-half}.

The second one-message constraint is
\begin{equation}
\label{eq:app-r-pia-linear}
r(q)\le\pi+a-\frac34.
\end{equation}
Using \(a=1-m\) and Eq.~\eqref{eq:app-r-alpha-delta}, Eq.~\eqref{eq:app-r-pia-linear} is equivalent to
\begin{equation}
\label{eq:app-pi-target}
\pi
\ge
\frac14+\frac{\alpha-\delta}{2}+m.
\end{equation}
We prove Eq.~\eqref{eq:app-pi-target} by splitting according to the size of the ordered corner separation.

If
\begin{equation}
\label{eq:app-case-one}
\alpha-\delta+2m\le\frac12,
\end{equation}
then the right-hand side of Eq.~\eqref{eq:app-pi-target} is at most \(1/2\). Since \(\pi\) is the larger of two parity weights whose sum is one, \(\pi\ge1/2\), and Eq.~\eqref{eq:app-pi-target} follows.

If instead
\begin{equation}
\label{eq:app-case-two}
\alpha-\delta+2m\ge\frac12,
\end{equation}
we use \(\pi\ge\alpha+\delta\). It suffices to show
\begin{equation}
\label{eq:app-alpha-delta-target}
\alpha+\delta
\ge
\frac14+\frac{\alpha-\delta}{2}+m,
\end{equation}
which is equivalent to
\begin{equation}
\label{eq:app-alpha-3delta}
\alpha+3\delta-2m\ge\frac12.
\end{equation}
The last inequality follows from the decomposition
\begin{equation}
\label{eq:app-alpha-3delta-decomp}
\alpha+3\delta-2m
=
(\alpha-\delta+2m)+4(\delta-m).
\end{equation}
The first term is at least \(1/2\) by Eq.~\eqref{eq:app-case-two}; the second is nonnegative because \(\delta\ge m\). This completes the proof of Eq.~\eqref{eq:app-r-pia-linear}.

Writing \(t=\Pi\oplus A=(\Pi+A)/2\), averaging Eqs.~\eqref{eq:app-r-pia-half} and \eqref{eq:app-r-pia-linear} gives
\begin{equation}
\label{eq:app-pioplusA-averaged}
R\le t,
\qquad
R\le 2t-\frac34.
\end{equation}
The lower envelope of these two constraints is \(F_{\Pi\oplus A}\), with the crossing at \(t=3/4\).

\subsection{Tightness and the four extremal sources}

The preceding subsections prove upper bounds. To see that they are exact frontiers, it remains to exhibit classical sources that attain every exposed segment. Four sources suffice:
\begin{equation}
\label{eq:app-tightness-values}
\begin{aligned}
\mathsf N &: (R,D,\Pi,A)=\left(\frac12,\frac14,\frac12,\frac34\right),\\
\mathsf E &: (R,D,\Pi,A)=\left(\frac23,\frac13,\frac23,1\right),\\
\mathsf B &: (R,D,\Pi,A)=\left(\frac34,\frac12,\frac12,1\right),\\
\mathsf F &: (R,D,\Pi,A)=\left(1,1,1,1\right).
\end{aligned}
\end{equation}
Here \(\mathsf N\) sends no information, so the posterior is uniform. The source \(\mathsf E\) sends a value known not to be the true value of \(x\), leaving a uniform posterior over the other three values. The source \(\mathsf B\) sends one bit of \(x\), for instance \(x_0\). The source \(\mathsf F\) sends \(x\) in full.

Flagged mixtures of classical sources preserve \(R,D,\Pi,A\) affinely, because the flag is part of the message. Therefore line segments between the points in Eq.~\eqref{eq:app-tightness-values} are themselves classically achievable.

For the four-state-discrimination frontier,
\begin{equation}
\label{eq:app-tight-D}
\mathsf N\to\mathsf E:
\ R=2D,
\qquad
\mathsf E\to\mathsf F:
\ R=\frac{1+D}{2}.
\end{equation}
These two mixtures cover the ranges \(1/4\le D\le1/3\) and \(1/3\le D\le1\), respectively.

For the parity frontier,
\begin{equation}
\label{eq:app-tight-pi}
\mathsf B\to\mathsf F:
\ R=\frac{1+\Pi}{2}.
\end{equation}
This covers the full range \(1/2\le\Pi\le1\).

For the four-state-discrimination-with-exclusion frontier, for the normalized task value \(D\oplus A=(D+A)/2\),
\begin{equation}
\label{eq:app-tight-DoplusA}
\mathsf N\to\mathsf F:
\ R=D\oplus A .
\end{equation}
This covers the range \(1/2\le D\oplus A\le1\).

For the parity-with-exclusion frontier, for the normalized task value \(\Pi\oplus A=(\Pi+A)/2\),
\begin{equation}
\label{eq:app-tight-PioplusA}
\begin{aligned}
\mathsf N\to\mathsf B &: R=2(\Pi\oplus A)-\frac34,\\
\mathsf B\to\mathsf F &: R=\Pi\oplus A .
\end{aligned}
\end{equation}
The first mixture covers \(5/8\le\Pi\oplus A\le3/4\), and the second covers \(3/4\le\Pi\oplus A\le1\). Every segment in Theorem~\ref{thm:classical-frontiers} is therefore attained by a classical source, completing the proof of exactness.

\subsection{Simultaneous affine bounds}
\label{app:master-affine-bounds}

The proof above gives more than the four single-assumption frontiers. The
six posterior inequalities survive averaging independently, and therefore
hold simultaneously for every classical strategy.

\begin{corollary}[Simultaneous affine bounds]
\label{cor:master-affine}
Every classical strategy satisfies
\begin{equation}
\label{eq:master-affine}
\begin{aligned}
R\leq\min\Bigl\{
&2D,\;
\frac{1+D}{2},\;
\frac{1+\Pi}{2},
\\[-1mm]
&\frac{D+A}{2},\;
\Pi+A-\frac34,\;
\frac{\Pi+A}{2}
\Bigr\},
\end{aligned}
\end{equation}
where all task values are evaluated in \(\mathcal C\). Each frontier in
Theorem~\ref{thm:classical-frontiers} is the tight lower envelope of the
corresponding entries after fixing the relevant source-task value.
\end{corollary}

\begin{proof}
For each posterior \(q\), the proof of
Theorem~\ref{thm:classical-frontiers} established the six inequalities
\[
r(q)\leq 2d(q),\quad
r(q)\leq \frac{1+d(q)}{2},\quad
r(q)\leq \frac{1+\pi(q)}{2},
\]
\[
r(q)\leq \frac{d(q)+a(q)}{2},\quad
r(q)\leq \pi(q)+a(q)-\frac34,\quad
r(q)\leq \frac{\pi(q)+a(q)}{2}.
\]
A classical source is an affine mixture of posteriors, as in
Eq.~\eqref{eq:app-affine}. Averaging the six inequalities over the
message \(\lambda\) gives Eq.~\eqref{eq:master-affine}. The tightness of
each single-assumption frontier follows from the endpoint mixtures in
Eqs.~\eqref{eq:app-tight-D}--\eqref{eq:app-tight-PioplusA}.
\end{proof}

Equation~\eqref{eq:master-affine} is useful when several source-task values are available simultaneously. The present work uses one
scalar source assumption at a time, so Theorem~\ref{thm:classical-frontiers}
uses only the corresponding face of this affine system.

\section{Noncommutative relaxation for the quantum region}
\label{app:quantum-relaxation}

This appendix specifies the outer approximation used in Fig.~\ref{fig:quantum-deviation}, following the noncommutative semidefinite-relaxation framework reviewed in Ref.~\cite{TavakoliPozasKerstjensBrownAraujo2024SDPReview}. It is an operator-family moment relaxation of the arbitrary-dimensional prepare-and-measure problem. State-weighted functionals are not assumed tracial.

\subsection{Operator families and moment matrices}

Let $\mathcal W_k$ be the reduced words of degree at most $k$ in the two binary RAC projectors $M_0,M_1$, with $M_y=M_y^\dagger=M_y^2$. For every positive operator family $F$ appearing in the optimization, define
\begin{equation}
\label{eq:app-qregion-functional}
L_F(w):=\Tr(Fw),
\end{equation}
and the moment matrix
\begin{equation}
\label{eq:app-qregion-moment}
\Gamma^F_{u,v}=L_F(u^\dagger v),
\qquad u,v\in\mathcal W_k.
\end{equation}
Every quantum realization gives $\Gamma^F\succeq0$ because $L_F(q^\dagger q)\ge0$ for all retained polynomials $q$. Hermiticity, projector reduction, and equality of shared reduced words are imposed entrywise. No cyclic relation $L_F(uv)=L_F(vu)$ is imposed on a state-weighted family.

The base positive families are the four preparations $\rho_x$ and their average $\bar\rho=\frac14\sum_x\rho_x$. They obey
\begin{equation}
\label{eq:app-qregion-normalization}
L_{\rho_x}(1)=1,
\qquad
L_{\bar\rho}(w)=\frac14\sum_xL_{\rho_x}(w)
\end{equation}
for every retained word.

\subsection{Source-task covers}

The task certificates are represented through positive cover slacks. For distinguishability,
\begin{equation}
\label{eq:app-qregion-D-cover}
\Gamma^{\sigma_D-\rho_x/4}\succeq0
\qquad\forall x.
\end{equation}
For parity,
\begin{equation}
\label{eq:app-qregion-P-covers}
\Gamma^{\sigma_\Pi-(\rho_{00}+\rho_{11})/4}\succeq0,
\qquad
\Gamma^{\sigma_\Pi-(\rho_{01}+\rho_{10})/4}\succeq0.
\end{equation}
For exclusion,
\begin{equation}
\label{eq:app-qregion-A-covers}
\Gamma^{\sigma_A-\frac14\sum_{x\ne z}\rho_x}\succeq0
\qquad\forall z.
\end{equation}
At the operator level, the minimum certificate traces equal the source-task values in Eq.~\eqref{eq:source-tasks-main}; the finite moment covers provide the outer relaxation used below. At $\Pi=1/2$, the parity family is eliminated and the equality $\rho_{00}+\rho_{11}=\rho_{01}+\rho_{10}$ is imposed on every retained moment.

\subsection{RAC objective and finite-level upper bounds}

The RAC objective is the linear functional
\begin{equation}
\label{eq:app-qregion-rac}
R=\frac18\sum_{x,y}
\begin{cases}
L_{\rho_x}(M_y),&x_y=0,\\
1-L_{\rho_x}(M_y),&x_y=1.
\end{cases}
\end{equation}
On each affine branch $F_T(t)=ct+b$, the relaxation maximizes $R-cT-b$ with the branch interval imposed on the corresponding source-task value: $\Tr\sigma_D$, $\Tr\sigma_\Pi$, $(\Tr\sigma_D+\Tr\sigma_A)/2$, or $(\Tr\sigma_\Pi+\Tr\sigma_A)/2$. Since every arbitrary-dimensional quantum realization induces feasible moment matrices satisfying Eqs.~\eqref{eq:app-qregion-moment}--\eqref{eq:app-qregion-rac}, the exact finite-level optimum is an upper bound on the quantum value.

The scan uses degree three, denoted $Q_3$, and is solved with CVXPY 1.8.2 and MOSEK 11. Points are accepted only when the primal--dual objective difference and all equality and positive-semidefinite residuals are at most $5\times10^{-6}$. For all four source assumptions, no gap between the BB84 feasible value and the corresponding branchwise upper bound larger than this tolerance is resolved. The grey markers are the same upper-bound relaxation evaluated on the displayed source-task grid.

\subsection{Finite-dimensional see-saw lower bounds}

The see-saw is used only to construct feasible points. For fixed states, $M_y$ is the Helstrom projector of the RAC contrast. For fixed projectors, the state update is an SDP over $\rho_x$ and the task certificates with an upper bound on the chosen source-task value. After convergence, the attained task value is recomputed by a separate discrimination or exclusion SDP for the final ensemble. The plotted $d=3$ triangles therefore use the optimized operational task value of the realized ensemble, rather than the trace of a nonoptimal certificate. They are feasible lower bounds and are not used to certify the grey upper envelope.

\section{Security model and post-measurement state}
\label{app:matched-security-state}

This appendix spells out the operator model behind
Sec.~\ref{sec:sdiqkd-operational-source}. The purpose is to make explicit
which objects are trusted, which objects are optimized over, and where the
retained-key state enters the entropy programs. The source assumption is
always imposed on the emitted ensemble after it has been lifted to a
Bob--Eve extension; no Hilbert-space dimension, purity, or internal
description of Bob's measurements is assumed.

\subsection{General prepare-and-measure realization}

For each value of $x$ \(x\in\{00,01,10,11\}\), the adversarial realization is
described by a normalized positive operator
\begin{equation}
\label{eq:app-security-states}
\varrho_x^{BE}\succeq0,
\qquad
\Tr\varrho_x^{BE}=1.
\end{equation}
This is the minimal object needed for the security proof: system \(B\) is
the system measured by Bob, and system \(E\) is retained by Eve. The
Hilbert spaces are unrestricted. Bob has three binary settings. With
\(B_y\) denoting the outcome-zero effect, a common Naimark dilation
allows us to write
\begin{equation}
\label{eq:app-security-Bob-projectors}
B_y=B_y^\dagger=B_y^2,
\qquad y=0,1,2.
\end{equation}
Equation~\eqref{eq:app-security-Bob-projectors} is only a dilation of
binary POVMs; it is not a compatibility assumption. No relation is
imposed between distinct settings, so \(B_yB_{y'}\) and \(B_{y'}B_y\)
remain different words in the noncommutative relaxation. Bob's algebra
commutes with Eve's algebra, as in a tensor-product realization.

The observed behavior is the complete twelve-entry table
\begin{equation}
\label{eq:app-security-full-table}
p_{xy}
=
\Tr\!\left[\varrho_x^{BE}(B_y\otimes I_E)\right],
\qquad
x\in\mathcal X,\quad y\in\{0,1,2\}.
\end{equation}
The first eight entries are the RAC-test entries. The four entries with
\(y=2\) constrain the same key measurement that appears in the retained
round. Keeping all twelve entries is important: replacing the complete
table by only the RAC score would enlarge the feasible Bob--Eve set and
would give a weaker security statement.

\subsection{Key map and cq state}

The retained branch is the fixed map
\begin{equation}
\label{eq:app-zkey-map}
0\mapsto00,
\qquad
1\mapsto11.
\end{equation}
After Bob performs setting \(2\), Eve's subnormalized state for Alice's
key value \(k\) and Bob's output \(b\) is
\begin{equation}
\label{eq:app-security-post-Eve}
\omega_{kb}^{E}
=
\Tr_B\!\left[(B_{b|2}\otimes I_E)
\varrho_{x_Z(k)}^{BE}\right],
\qquad
B_{0|2}=B_2,\quad B_{1|2}=I_B-B_2 .
\end{equation}
This definition is forced by the measurement post-processing: Bob's
classical outcome is kept, and the post-measurement quantum system on
Eve's side is traced over \(B\). The retained-round state is therefore
\begin{equation}
\label{eq:app-security-KKBE}
\omega_{K K_B E}
=
\frac12\sum_{k,b=0}^{1}
\ketbra{k}{k}\otimes\ketbra{b}{b}\otimes\omega_{kb}^{E}.
\end{equation}
Taking the classical marginal gives
\begin{equation}
\label{eq:app-security-QZ}
Q_Z
=
\frac12\Tr\omega_{0,1}^{E}
+
\frac12\Tr\omega_{1,0}^{E}
=
\frac12(1-p_{00,2}+p_{11,2}).
\end{equation}
The first term is the event \(K=0,K_B=1\); the second is the event
\(K=1,K_B=0\). The last equality uses the outcome-zero convention in
Eq.~\eqref{eq:app-security-full-table}. Tracing out Bob's classical
register gives
\begin{equation}
\label{eq:app-security-omegaKE}
\omega_{KE}
=
\frac12\ketbra{0}{0}\otimes\rho_0^E
+
\frac12\ketbra{1}{1}\otimes\rho_1^E,
\end{equation}
where
\[
\rho_0^E=\Tr_B\varrho_{00}^{BE},
\qquad
\rho_1^E=\Tr_B\varrho_{11}^{BE}.
\]
Since the key register is classical, the conditional entropy has the
block form
\begin{equation}
\label{eq:app-security-HKE}
H(K|E)
=
1+\frac12H(\rho_0^E)+\frac12H(\rho_1^E)
-H\!\left(\frac{\rho_0^E+\rho_1^E}{2}\right).
\end{equation}
Equation~\eqref{eq:app-security-HKE} is exact. The min-entropy and
PM-BFF programs below lower-bound the minimum of this quantity over the
same admissible source and measurement operators.

\subsection{Operator form of the source assumptions}

The optimized source tasks are imposed through their dual covers. For
distinguishability, a single certificate \(\sigma_D\) must dominate each
weighted preparation:
\begin{equation}
\label{eq:app-security-D-covers}
\sigma_D-\frac14\varrho_x^{BE}\succeq0
\qquad\forall x.
\end{equation}
Taking the trace and minimizing \(\Tr\sigma_D\) is precisely the dual of
the four-state discrimination problem. For parity, the two effective
parity preparations are the even and odd mixtures, so the two covers are
\begin{equation}
\label{eq:app-security-P-covers}
\sigma_\Pi-\frac14(\varrho_{00}^{BE}+\varrho_{11}^{BE})\succeq0,
\qquad
\sigma_\Pi-\frac14(\varrho_{01}^{BE}+\varrho_{10}^{BE})\succeq0.
\end{equation}
For exclusion, the reported value \(z\) is successful on all
preparations except \(z\). The corresponding dual cover is
\begin{equation}
\label{eq:app-security-A-covers}
\sigma_A-\frac14\sum_{x\neq z}\varrho_x^{BE}\succeq0
\qquad\forall z.
\end{equation}
The scalar source restrictions are therefore
\begin{equation}
\label{eq:app-security-task-traces}
\begin{aligned}
\Tr\sigma_D&\le\tau_D,
&
\Tr\sigma_\Pi&\le\tau_\Pi,\\
\Tr\sigma_D+\Tr\sigma_A&\le2\tau_{D\oplus A},
&
\Tr\sigma_\Pi+\Tr\sigma_A&\le2\tau_{\Pi\oplus A}.
\end{aligned}
\end{equation}
The composite conditions are trace-sum conditions only. They do not
impose separate bounds on \(D\) and \(A\), or on \(\Pi\) and \(A\).
This is the operator form of the operational assumption used in the main
text.

At the parity-oblivious endpoint \(\Pi=1/2\), both positive slacks in
Eq.~\eqref{eq:app-security-P-covers} have zero trace. A positive
operator with zero trace is the zero operator, so the two parity covers
collapse to the exact identity
\begin{equation}
\label{eq:app-security-exact-parity}
\varrho_{00}^{BE}+\varrho_{11}^{BE}
=
\varrho_{01}^{BE}+\varrho_{10}^{BE}.
\end{equation}
In the numerical hierarchy, Eq.~\eqref{eq:app-security-exact-parity} is
imposed on every retained moment, including the independent
three-setting Bob moments. This avoids replacing an exact source
condition by a relaxation-dependent parity cover.

\section{Min-entropy noncommutative polynomial optimization}
\label{app:minentropy-sdp}

This appendix gives the exact operator problem behind the min-entropy
certificate and the finite noncommutative relaxation used to upper-bound
Eve's retained-key guessing probability. The direction is important:
the SDP maximizes Eve's guessing probability over an outer
approximation, so the result is an upper bound on
\(p_{\rm guess}\) and hence a lower bound on \(H_{\min}(K|E)\).

\subsection{Exact operator problem}

Let \(E\) denote Eve's effect for the guess \(K=0\). After a Naimark
dilation it is enough to take
\begin{equation}
\label{eq:app-min-E-relations}
E=E^\dagger=E^2,
\qquad
[E,B_y]=0
\quad(y=0,1,2).
\end{equation}
The commutation in Eq.~\eqref{eq:app-min-E-relations} is only the
Bob--Eve tensor-product commutation. It does not impose any commutation
between Bob's settings. For fixed feasible states, Eve's probability of
guessing the retained key is
\begin{equation}
\label{eq:app-min-guess-functional}
g(E)
:=
\frac12\Tr(\varrho_{00}^{BE}E)
+\frac12\Tr\!\left[\varrho_{11}^{BE}(I-E)\right].
\end{equation}
The first term is the probability of guessing \(0\) on the \(00\) branch;
the second is the probability of guessing \(1\) on the \(11\) branch.
The exact optimization is
\begin{equation}
\label{eq:app-min-exact-NCPOP}
\begin{aligned}
p_{{\rm guess},T}^{\star}(\mathbf p,\tau)
&=\sup\ g(E)
\\
\mathrm{s.t.}\quad&
\varrho_x^{BE}\succeq0,
\quad\Tr\varrho_x^{BE}=1,
\\
& B_y=B_y^\dagger=B_y^2,\\
& E=E^\dagger=E^2,
\qquad [E,B_y]=0,
\\
& \Tr(\varrho_x^{BE}B_y)=p_{xy}
\quad\forall x,y,
\\
& \text{the covers and trace bound defining }T\le\tau .
\end{aligned}
\end{equation}
No dimension, purity, or commutation relation among Bob's settings is
assumed. If the finite relaxation returns a validated upper bound
\(\overline p_{{\rm guess},T}\), the certified min-entropy rate is
\begin{equation}
\label{eq:app-min-rate}
r_T
\ge
r_T^{\min}
:=
-\log_2\overline p_{{\rm guess},T}-h_2(Q_Z).
\end{equation}

\subsection{Reduced word algebra}

Let \(\mathcal A_{\min}\) be the unital \(*\)-algebra generated by
\(B_0,B_1,B_2,E\), reduced only by
\begin{equation}
\label{eq:app-min-algebra-relations}
B_y^2=B_y,
\qquad
E^2=E,
\qquad
EB_y=B_yE.
\end{equation}
Different Bob settings are never reordered. The degree-three row set is
\begin{equation}
\label{eq:app-min-word-set}
\mathcal W_{\min}^{(3)}
=
\{\operatorname{red}(w):w\in\{B_0,B_1,B_2,E\}^{\ast},\ |w|\le3\}.
\end{equation}
The production implementation verifies that this set contains \(32\)
reduced words. Its product closure contains \(284\) reduced words,
grouped into \(164\) adjoint orbits. The Bob-only subset contains
\(22\) words. These counts are a structural check on the relaxation:
changing them changes the SDP.

For each positive family
\begin{equation}
\label{eq:app-min-families}
F\in
\{\bar\varrho,\varrho_{00},\varrho_{01},
\varrho_{10},\varrho_{11},\sigma_D,\sigma_\Pi,\sigma_A\},
\end{equation}
when present, define the state-weighted functional
\begin{equation}
\label{eq:app-min-functional}
L_F(w)=\Tr(Fw).
\end{equation}
Its moment matrix is
\begin{equation}
\label{eq:app-min-Gamma}
\Gamma^F_{u,v}=L_F(u^\dagger v),
\qquad u,v\in\mathcal W_{\min}^{(3)},
\end{equation}
and positivity of \(F\) gives
\begin{equation}
\label{eq:app-min-Gamma-psd}
\Gamma^F\succeq0.
\end{equation}
The state normalizations and the average-state identity are
\begin{equation}
\label{eq:app-min-normalization-average}
L_{\varrho_x}(1)=1,
\qquad
L_{\bar\varrho}(w)=\frac14\sum_xL_{\varrho_x}(w)
\quad\forall w\in\mathcal C_{\min}^{(3)},
\end{equation}
where \(\mathcal C_{\min}^{(3)}=\{u^\dagger v:u,v\in
\mathcal W_{\min}^{(3)}\}\). The functional \(L_{\bar\varrho}\) is a
state-weighted trace, not a tracial state. Consequently, no cyclic
relations such as \(L_{\bar\varrho}(uv)=L_{\bar\varrho}(vu)\) are
imposed.

\subsection{Lifted source covers and observed behavior}

The source covers from Appendix~\ref{app:matched-security-state} are
lifted by replacing each operator by its moment matrix. Thus,
\begin{equation}
\label{eq:app-min-lifted-covers}
\begin{aligned}
\Gamma^{\sigma_D}-\frac14\Gamma^{\varrho_x}&\succeq0,
\\
\Gamma^{\sigma_\Pi}
-\frac14(\Gamma^{\varrho_{00}}+\Gamma^{\varrho_{11}})&\succeq0,
\\
\Gamma^{\sigma_\Pi}
-\frac14(\Gamma^{\varrho_{01}}+\Gamma^{\varrho_{10}})&\succeq0,
\\
\Gamma^{\sigma_A}-\frac14\sum_{x\neq z}\Gamma^{\varrho_x}&\succeq0.
\end{aligned}
\end{equation}
Only the covers required by the chosen task are active. For a composite
task, the only scalar restriction is the trace-sum condition in
Eq.~\eqref{eq:app-security-task-traces}; no individual component cap is
added.

The complete behavior is imposed as first-order moments,
\begin{equation}
\label{eq:app-min-full-table}
L_{\varrho_x}(B_y)=p_{xy}
\qquad\forall x,\ y=0,1,2.
\end{equation}
The average-setting moments follow from
Eq.~\eqref{eq:app-min-normalization-average}. The relaxation objective is
\begin{equation}
\label{eq:app-min-moment-objective}
\overline p_{{\rm guess},T}
=
\max
\left
\{
\frac12L_{\varrho_{00}}(E)
+
\frac12[1-L_{\varrho_{11}}(E)]
\right\}.
\end{equation}
Since the truncated moment set is an outer approximation of the exact
operator set, Eq.~\eqref{eq:app-min-moment-objective} gives an upper
bound on Eve's physical guessing probability. Substitution into
Eq.~\eqref{eq:app-min-rate} gives the finite-level min-entropy
certificate.

\section{Prepare-and-measure Brown--Fawzi--Fawzi optimization}
\label{app:pm-bff}

This appendix gives the PM-BFF relaxation used for the conditional
von Neumann entropy curves. The implementation is task-generic: the same
64-word PM-BFF basis, the same complete three-setting Bob degree-three
block, the same sharp localizers, and the same Eve-only Sylvester
equations are used for \(D\), \(\Pi\), \(D\oplus A\), and
\(\Pi\oplus A\). The source task changes only the dual-cover families
and the trace bound. At the parity-oblivious endpoint \(\Pi=1/2\), the
parity covers are replaced by the exact moment equality of
Eq.~\eqref{eq:app-security-exact-parity}.

\subsection{Node functional, stationarity, and sharp norm}

For the retained-key cq state, set
\(\eta_k=\rho_k^E\) and \(\Omega=\eta_0+\eta_1\). At an active
quadrature node \(0<t<1\), the BFF variational objective becomes
\begin{equation}
\label{eq:app-bff-node-functional}
\begin{aligned}
\beta_t
=\inf_{Z_0,Z_1}\frac12\sum_{k=0}^1\Bigl\{
&\Tr[\eta_k(Z_k+Z_k^\dagger)]\\
&+(1-t)\Tr(\eta_kZ_k^\dagger Z_k)\\
&+t\Tr(\Omega Z_kZ_k^\dagger)
\Bigr\}.
\end{aligned}
\end{equation}
The two variables \(Z_0,Z_1\) act only on Eve's system. The first two
terms are branchwise; the last term contains the common marginal
\(\Omega\) and couples the two branches. Varying with respect to
\(Z_k^\dagger\) gives the Sylvester equation
\begin{equation}
\label{eq:app-bff-Sylvester}
\eta_k+(1-t)Z_k\eta_k+t\Omega Z_k=0,
\end{equation}
with the adjoint equation obtained by conjugation.

\begin{lemma}[Sharp auxiliary norm]
\label{lem:app-bff-sharp-norm}
Every solution of Eq.~\eqref{eq:app-bff-Sylvester} on the support of
\(\Omega\) obeys
\begin{equation}
\label{eq:app-bff-sharp-norm}
\|Z_k\|\le\frac{1}{2\sqrt{t(1-t)}}.
\end{equation}
\end{lemma}

\begin{proof}
Let \(Z_kv=su\) and \(Z_k^\dagger u=sv\) be singular vectors for
\(s=\|Z_k\|\). Put
\(a=\langle u|\Omega|u\rangle\) and
\(b=\langle v|\eta_k|v\rangle\). Taking the \(u,v\) matrix element of
Eq.~\eqref{eq:app-bff-Sylvester} gives
\[
s[ta+(1-t)b]
=
|\langle u|\eta_k|v\rangle|.
\]
Because \(0\preceq\eta_k\preceq\Omega\), the right-hand side is at most
\(\sqrt{ab}\). The arithmetic-geometric mean inequality gives
\(ta+(1-t)b\ge2\sqrt{t(1-t)ab}\), hence
\(s\le[2\sqrt{t(1-t)}]^{-1}\) whenever \(ab>0\). The remaining support
cases follow by restriction to the support of \(\Omega\) and continuity.
\end{proof}

For every Eve-side test word \(A\), the weak form used in the moment
relaxation is
\begin{equation}
\label{eq:app-bff-weak-Sylvester}
\Tr(\eta_kA^\dagger)
+(1-t)\Tr(\eta_kA^\dagger Z_k)
+t\Tr(\Omega Z_kA^\dagger)=0.
\end{equation}
No Bob word is allowed in \(A\). Including Bob-conditioned tests would
impose a stronger equation than the Eve-side BFF stationarity condition
and is therefore not part of the certificate.

The right-endpoint Gauss--Radau contribution is absorbed analytically.
The active interior nodes have positive coefficients
\(\alpha_i=w_i/(t_i\ln2)\), and the finite-node lower bound has the form
\begin{equation}
\label{eq:app-bff-quadrature}
H(K|E)\ge\sum_i\alpha_i(1+\beta_{t_i}).
\end{equation}
Optimizing each node independently is conservative because
\begin{equation}
\label{eq:app-bff-independent-nodes}
\sum_i\alpha_i\inf_{\mathcal F_T}(1+\beta_{t_i})
\le
\inf_{\mathcal F_T}\sum_i\alpha_i(1+\beta_{t_i}).
\end{equation}
This is the certification direction used in the production curves.

\subsection{Noncommutative operator algebra and exact word set}

Write \(z_k=Z_k\) and \(d_k=Z_k^\dagger\). The algebra is generated by
\(B_0,B_1,B_2,z_0,d_0,z_1,d_1\), with
\begin{equation}
\label{eq:app-bff-algebra}
B_y^2=B_y,
\qquad [B_y,z_k]=[B_y,d_k]=0,
\qquad d_k=z_k^\dagger.
\end{equation}
Distinct Bob settings are never reordered, and no commutation is imposed
between the two Eve branches. The implemented PM-BFF word set is stated
explicitly. First define the two Bob-word sets
\[
\mathcal B_8=
\{1,B_0,B_1,B_2,B_0B_1,B_1B_0,B_0B_1B_0,B_1B_0B_1\},
\]
and
\[
\mathcal B_6=
\{1,B_0,B_1,B_2,B_0B_1,B_1B_0\}.
\]
The PM-BFF rows consist of the words in \(\mathcal B_8\); the same
words followed by one Eve symbol \(z_k\) or \(d_k\); and the words in
\(\mathcal B_6\) followed by the two-symbol suffix \(d_kz_k\) or
\(z_kd_k\). Equivalently,
\begin{equation}
\label{eq:app-bff-exact-basis}
\begin{aligned}
\mathcal W_{\rm BFF}
={}&
\mathcal B_8
\cup
\bigcup_{k=0}^{1}\bigcup_{s\in\{z_k,d_k\}}
\{ws:w\in\mathcal B_8\}
\cup{}
\bigcup_{k=0}^{1}\bigcup_{s\in\{d_kz_k,z_kd_k\}}
\{ws:w\in\mathcal B_6\}.
\end{aligned}
\end{equation}
Here \(ws\) denotes the word obtained by appending the suffix \(s\)
to the Bob word \(w\). Thus the row count is \(8+2\cdot2\cdot8+2\cdot2\cdot6=64\). The
product closure of this basis contains \(1458\) reduced words, grouped
into \(761\) adjoint orbits. These numbers specify the relaxation used
in the production run. For every positive family \(F\),
\begin{equation}
\label{eq:app-bff-main-moment}
\Gamma^F_{u,v}=L_F(u^\dagger v),
\qquad u,v\in\mathcal W_{\rm BFF},
\qquad \Gamma^F\succeq0.
\end{equation}

The five base positive families are
\(\bar\varrho,\varrho_{00},\varrho_{01},\varrho_{10},\varrho_{11}\).
They satisfy
\begin{equation}
\label{eq:app-bff-average-all-words}
L_{\varrho_x}(1)=L_{\bar\varrho}(1)=1,
\qquad
L_{\bar\varrho}(w)=\frac14\sum_xL_{\varrho_x}(w)
\end{equation}
for every retained word. As in the min-entropy relaxation,
\(L_{\bar\varrho}\) is not treated as a tracial state.

\subsection{Source covers and sharp localizers}

Every source cover is represented by a positive slack. For example,
\begin{equation}
\label{eq:app-bff-cover-slacks}
\Delta_{D,x}=\sigma_D-\frac14\varrho_x,
\qquad
\Delta_{A,z}=\sigma_A-\frac14\sum_{x\ne z}\varrho_x.
\end{equation}
The parity covers are defined analogously when the parity bound is not
the exact endpoint \(\Pi=1/2\). Moment positivity and norm localizers are
imposed on every active slack.

A separate PSD block for \(\sigma_D,\sigma_\Pi\), or \(\sigma_A\) is not
needed. Choose one right-hand side \(G_a\) appearing in a cover
\(\sigma\succeq G_a\). The certificate moments are reconstructed by
\begin{equation}
\label{eq:app-bff-certificate-reconstruction}
L_\sigma(w)=L_{\sigma-G_a}(w)+L_{G_a}(w).
\end{equation}
Both terms on the right are positive functionals, so the positivity of
\(\sigma\) follows. This saves a redundant block while keeping all
moments of the certificate available for trace bounds and shared covers.

Let \(c_t=[2\sqrt{t(1-t)}]^{-1}\). If a suffix \(s\) has length
\(\ell\), the sharp norm bound gives \(\|s\|\le c_t^\ell\). Equivalently,
the localizer can be written with the coefficient \([4t(1-t)]^\ell\):
\begin{equation}
\label{eq:app-bff-localizer}
L_F(q^\dagger q)-[4t(1-t)]^\ell L_F(q^\dagger s^\dagger sq)\ge0.
\end{equation}
The production basis has eight suffix sectors per positive family or
cover slack: four length-one sectors of size \(8\times8\), and four
length-two sectors of size \(6\times6\). These are the sharp-cq
localizers generated by Lemma~\ref{lem:app-bff-sharp-norm}.

\subsection{Complete physical hierarchy for Bob's settings}

The entropy basis is supplemented by the complete Bob degree-three word
set
\begin{equation}
\label{eq:app-bff-Bob3-words}
\mathcal W_{B,3}
=
\{\operatorname{red}(w):w\in\{B_0,B_1,B_2\}^{\ast},\ |w|\le3\}.
\end{equation}
This set has \(22\) rows. Its product closure contains \(190\) reduced
Bob words, grouped into \(106\) adjoint orbits. The \(89\) orbit
representatives not already present in the PM-BFF closure are introduced
as extra variables, giving \(164\) additional real scalar variables per
positive family. For every positive family or cover slack,
\begin{equation}
\label{eq:app-bff-Bob3-block}
\Gamma^{F,B3}_{u,v}=L_F(u^\dagger v),
\qquad u,v\in\mathcal W_{B,3},
\qquad \Gamma^{F,B3}\succeq0.
\end{equation}
This block forces all three Bob settings, including the key setting
\(B_2\), to belong to one common physical realization. The Bob-block
moments shared with Eq.~\eqref{eq:app-bff-main-moment} are identified,
so the entropy hierarchy and the Bob hierarchy cannot assign different
values to the same word.

\subsection{Weak Sylvester system and node objective}

The weak Sylvester equations are imposed only on Eve-side test words
visible in the product closure. For the \(K=0\) branch and the \(K=1\)
branch the implemented equations are
\begin{equation}
\label{eq:app-bff-Sylvester-functional}
\begin{aligned}
0={}&L_{\varrho_{00}}(A^\dagger)
 +(1-t)L_{\varrho_{00}}(A^\dagger z_0)
 +tL_{\varrho_{00}+\varrho_{11}}(z_0A^\dagger),
\\
0={}&L_{\varrho_{11}}(A^\dagger)
 +(1-t)L_{\varrho_{11}}(A^\dagger z_1)
 +tL_{\varrho_{00}+\varrho_{11}}(z_1A^\dagger).
\end{aligned}
\end{equation}
The production closure gives \(11\) complex equations per branch, hence
\(22\) complex equations, or \(44\) real scalar equalities, per active
node. The test directions include \(A=z_k\), which gives an independent
reconstruction of the corresponding branch contribution.

The node objective is
\begin{equation}
\label{eq:app-bff-node-objective}
\begin{aligned}
\beta_t={}&\frac12\Bigl[
L_{\varrho_{00}}(z_0)+L_{\varrho_{00}}(d_0)
 +(1-t)L_{\varrho_{00}}(d_0z_0)
\\
&\hspace{26mm}
+tL_{\varrho_{00}+\varrho_{11}}(z_0d_0)
\Bigr]
\\
&+\frac12\Bigl[
L_{\varrho_{11}}(z_1)+L_{\varrho_{11}}(d_1)
 +(1-t)L_{\varrho_{11}}(d_1z_1)
\\
&\hspace{26mm}
+tL_{\varrho_{00}+\varrho_{11}}(z_1d_1)
\Bigr].
\end{aligned}
\end{equation}
This is exactly Eq.~\eqref{eq:app-bff-node-functional} written in
moment form for the two retained-key branches \(00\) and \(11\).

\subsection{Block counts and numerical validation}

Each positive family or source-cover slack carries three kinds of
constraints: the \(64\times64\) PM-BFF moment block, the complete
\(22\times22\) Bob degree-three block, and the sharp localizers generated
by the retained auxiliary suffixes. The same 64-word PM-BFF basis, Bob
block, localizer sectors, and Eve-only Sylvester equations are used for
source depolarization, receiver-side depolarization, direct-sum leakage,
the fixed assumption \(D\leq1/2\), and the \(\nu_{\rm M}=1\) endpoints.
The source task changes only the cover slacks and the trace constraint.
At \(\Pi=1/2\), Eq.~\eqref{eq:app-security-exact-parity} replaces the
parity certificate and its two covers.

For a task with \(c\) source covers, the program contains \(5+c\)
positive block families: the five physical state families and the
\(c\) source-cover slacks. Each such family contributes one main PM-BFF
block, one Bob degree-three block, and eight sharp localizer blocks. The
resulting block counts are listed in
Table~\ref{tab:app-pmbff-block-counts}. The parity rows distinguish
the general parity-cover case from the exact endpoint \(\Pi=1/2\) used
in the parity-oblivious runs. These counts are generated from the same
task constructor that produces the numerical programs and are
independently checked before production runs.

\begin{table*}[t]
\begingroup
\footnotesize
\begin{ruledtabular}
\begin{tabular}{lrrrrrrrrr}
source assumption & dual fam. & covers & pos. fam. & main & local. & Bob & PSD & PSD scalars & moment scalars\\
\hline
$D$ & 1 & 4 & 9 & 9 & 72 & 9 & 90 & 90,918 & 9,732 \\
$\Pi$ cover, $\Pi>1/2$ & 1 & 2 & 7 & 7 & 56 & 7 & 70 & 70,714 & 9,732 \\
exact $\Pi=1/2$ & 0 & 0 & 5 & 5 & 40 & 5 & 50 & 50,510 & 8,110 \\
$D\oplus A$ & 2 & 8 & 13 & 13 & 104 & 13 & 130 & 131,326 & 11,354 \\
$\Pi\oplus A$ & 2 & 6 & 11 & 11 & 88 & 11 & 110 & 111,122 & 11,354 \\
\end{tabular}
\end{ruledtabular}
\caption{Task-dependent block counts for the PM-BFF certificate. ``dual fam.''
counts auxiliary dual certificate families; ``covers'' counts positive
source-cover slacks; and ``pos. fam.'' counts the base preparation and
average families together with the active cover slacks. Each positive
family carries one \(64\times64\) PM-BFF moment block, one \(22\times22\)
Bob degree-three block, and the listed sharp-localizer blocks. ``PSD
scalars'' gives the real scalarized size of all PSD cones; ``moment
scalars'' counts the independent real moment variables from the PM-BFF
moments and the extra Bob degree-three moments. At the exact
parity-oblivious endpoint, the equality
Eq.~\eqref{eq:app-security-exact-parity} replaces the parity certificate
and its two covers.}
\label{tab:app-pmbff-block-counts}
\endgroup
\end{table*}

At every node, the reported lower-bound contribution is the minimum of
the solver primal objective, solver dual objective, and reconstructed
objective, shifted downward by \(10^{-6}\). A node is accepted only when
the solver status is optimal and independent reconstruction verifies
state normalization, average-state identities, the complete observed
table, task trace bounds, source-cover positivity, localizer positivity,
Bob-block positivity, Sylvester residuals, and objective consistency to
tolerance \(5\times10^{-6}\). The guarded value is converted to
\(\Delta H_i=\alpha_i(1+\beta_i)\), and the \(34\) active node
contributions are summed. Equation~\eqref{eq:app-bff-independent-nodes}
is the reason this nodewise procedure gives a valid entropy lower bound.

\section{Numerical protocol and critical-point extraction}
\label{app:keyrate-numerics}

The operator programs are specified in the preceding appendices. This
section records the numerical parameters, validation criteria, and
critical-point convention used for the displayed finite-level bounds.
The security calculations use the complete three-setting behavior
\(\mathbf p\) of Eq.~\eqref{eq:security-behavior}. For the plotted
families, the honest RAC entries and operational task values are those
in Eqs.~\eqref{eq:bb84-visibility-task-values},
\eqref{eq:joint-source-receiver-rac-score}, and
\eqref{eq:bb84-leakage-task-values}, with the retained branch fixed by
Eq.~\eqref{eq:key-branch} and \(Q_Z\) evaluated by
Eq.~\eqref{eq:key-qber}. Every source assumption, including each leakage
task, is optimized independently.

For the source-depolarized family, the parameter \(\nu\) specifies
Alice's emitted ensemble and therefore the operational source bound. The
receiver-side parameter \(\nu_{\rm M}\) depolarizes the reference
decoders used to generate the observed table after the source family is
chosen. Hence receiver-side depolarization changes the behavior
\(\mathbf p\), but it does not change the source-task value used as the
bound. For direct-sum leakage, the leakage parameter specifies the
emitted ensemble and the key-setting branch is error free in the honest
table.

The min-entropy implementation uses the \(32\)-word reduced
degree-three basis in Eq.~\eqref{eq:app-min-word-set}, its \(284\)-word
product closure and \(164\) adjoint orbits, and the complete observed
table. Its Bob-only subset has \(22\) words. The same basis, reduction
rules, source-cover generator, guessing objective, and acceptance tests
are used for every min-entropy curve and every min-entropy critical
point in the manuscript.

The PM-BFF implementation uses the 64-word basis in
Eq.~\eqref{eq:app-bff-exact-basis}, its \(1458\)-word product closure,
\(761\) adjoint orbits, the complete \(22\)-row Bob degree-three block of
Eq.~\eqref{eq:app-bff-Bob3-words}, eight sharp localizers per positive
family or source-cover slack, and \(11\) complex Eve-only Sylvester
equations per key branch. The right-endpoint Gauss--Radau term is
absorbed analytically, leaving exactly \(34\) active interior nodes
\(0<t_i<1\). This identical word basis, Bob block, localizer construction,
Sylvester system, and \(34\)-node quadrature are used for every PM-BFF
curve and every PM-BFF critical point. The source task changes only its
cover constraints and trace bound, giving the task-dependent block
counts in Table~\ref{tab:app-pmbff-block-counts}. At \(\Pi=1/2\), the
exact retained-moment parity equality replaces the auxiliary parity
certificate and its two covers. Independent static verification
reproduces the word-basis, closure, orbit, localizer, Bob-block,
Sylvester, and task-dependent block counts before numerical production.

The semidefinite programs are solved with CVXPY~1.8.2 and MOSEK~11. A
point is retained only when the solver status is optimal and independent
reconstruction verifies state normalization, average-state identities,
the complete observed table, task trace bounds, source-cover
positivity, localizer positivity, Bob-block positivity, all Sylvester
equations where present, and objective consistency to tolerance
\(5\times10^{-6}\). The PM-BFF node contribution is the minimum of the
primal, dual, and reconstructed objectives, shifted downward by
\(10^{-6}\); the min-entropy guessing-probability bound is the dual
objective shifted upward by \(10^{-6}\). The final key-rate lower bound
is recomputed from the guarded entropy contribution and \(h_2(Q_Z)\).
Every displayed marker is therefore an accepted finite-level certificate
generated by the fixed implementation specified above for its certificate
type.

The certification is numerical at the stated tolerance; interval
arithmetic and rational reconstruction are not used. A critical
visibility quoted in the main text is a local root estimate constrained
by the nearest accepted samples with nonpositive and positive key-rate
lower bounds. For the source-only \(D^{(\nu)}\) curves, the PM-BFF
bracket is \([0.72,0.74]\) and the min-entropy bracket is
\([0.81,0.82]\). A linear root estimate from the nearest accepted
samples gives \(0.7202456\) for PM-BFF and \(0.8161951\) for
min-entropy; the displayed values are \(0.7202\) and \(0.8162\). The
receiver-side-depolarized points use the same nearest-bracket root-estimate convention. Every root summarizes a crossing and is not
treated as an additional SDP evaluation. The direct-sum leakage lower
bounds are positive at every sampled \(\ell<1\) and reproduce the
zero-rate endpoint at \(\ell=1\) within the same tolerance.

\section{Affine PM-BFF certificates as finite-key inputs}
\label{app:finite-key-geat}

This appendix records what part of the PM-BFF calculation can be used in
a future finite-key analysis. The manuscript reports asymptotic i.i.d.
rates; the point proved here is narrower. An accepted PM-BFF dual
certificate gives an affine one-round entropy bound of the form required
by generalized entropy accumulation. Numerical finite-\(n\) rates require
a specified parameter-estimation test, error-correction transcript,
privacy-amplification map, and smoothing parameters.

\begin{lemma}[Affine lower bounds from PM-BFF dual certificates]
\label{lem:finite-key-affine-pmbff}
Fix a source assumption \(T\), a source bound \(\tau\), and an active
BFF node \(t_i\). Any feasible dual solution of the finite PM-BFF node
relaxation defines an affine function
\(\ell_i(\mathbf p,\tau)\) such that, for every physical one-round
Bob--Eve realization compatible with the complete behavior \(\mathbf p\)
and the lifted source bound \(T\leq\tau\),
\begin{equation}
\label{eq:finite-key-node-affine}
\ell_i(\mathbf p,\tau)
\leq
\beta_{T,i}^{\star}(\mathbf p,\tau).
\end{equation}
Consequently,
\begin{equation}
\label{eq:finite-key-affine-entropy}
f_T(\mathbf p,\tau)
:=
\sum_i\alpha_i\bigl[1+\ell_i(\mathbf p,\tau)\bigr]
\leq
\mathsf H_T(\mathbf p,\tau).
\end{equation}
\end{lemma}

\begin{proof}
For fixed \(t_i\), the PM-BFF node program is a minimization of the
linear moment form in Eq.~\eqref{eq:app-bff-node-objective}, subject to
linear equalities fixing the complete behavior \(\mathbf p\), the source
trace bound \(T\leq\tau\), and positive-semidefinite moment and localizer
constraints. Its finite relaxation is an outer approximation of the
physical one-round feasible set. Therefore its optimum is no larger than
the exact physical infimum \(\beta_{T,i}^{\star}(\mathbf p,\tau)\).

A feasible dual solution has constraints independent of the numerical
values of \(\mathbf p\) and \(\tau\); these values enter only through the
dual objective, which is affine in \((\mathbf p,\tau)\). Weak duality gives
an affine lower bound on the relaxed node optimum for every admissible
\((\mathbf p,\tau)\), and hence Eq.~\eqref{eq:finite-key-node-affine}.
Multiplying by the positive quadrature weights \(\alpha_i\), summing over
the active nodes, and using Proposition~\ref{prop:pmbff-entropy-certificate}
gives Eq.~\eqref{eq:finite-key-affine-entropy}.
\end{proof}

The function \(f_T\) is the one-round entropy lower bound used as a
min-tradeoff function in entropy-accumulation arguments. The next
statement records the conditional reduction, without choosing the
finite-size constants.

\begin{proposition}[GEAT input from affine certificates]
\label{prop:finite-key-reduction}
Consider a specified sequential implementation satisfying the hypotheses
of a generalized entropy-accumulation theorem. In each round, Alice's
emitted ensemble obeys the same scalar source bound \(T\leq\tau\), Bob's
setting choice and public announcements define a complete one-round
behavior, and Eve may keep arbitrary quantum memory subject to the causal
ordering of the rounds. If the parameter-estimation test accepts only
behaviors on which the affine bound \(f_T\) of
Eq.~\eqref{eq:finite-key-affine-entropy} is valid, then \(f_T\) is an
admissible min-tradeoff function. The chosen entropy-accumulation theorem
therefore gives a bound of the form
\begin{equation}
\label{eq:finite-key-geat-form}
H_{\min}^{\varepsilon_s}(K^n|E\mathsf C)_{|\Omega}
\geq
n f_T(\bar{\mathbf p},\tau)
-
\Delta_{\rm GEAT}(n,\varepsilon_s,\varepsilon_{\rm PE},f_T),
\end{equation}
where \(\Omega\) is the acceptance event, \(\mathsf C\) denotes the public
classical transcript, \(\bar{\mathbf p}\) is the accepted empirical
behavior, and \(\Delta_{\rm GEAT}\) is the finite-size correction given
by the chosen generalized entropy-accumulation theorem.
After one-way error correction with leakage \(\mathrm{leak}_{\rm EC}\)
and two-universal privacy amplification, the corresponding finite-key
analysis certifies any key length satisfying
\begin{equation}
\label{eq:finite-key-length}
\ell
\leq
H_{\min}^{\varepsilon_s}(K^n|E\mathsf C)_{|\Omega}
-
\mathrm{leak}_{\rm EC}
-
2\log_2\frac{1}{\varepsilon_{\rm PA}}
\end{equation}
with the smoothing, correctness, and privacy-amplification errors of the
chosen finite-key analysis.
\end{proposition}

\begin{proof}
By Lemma~\ref{lem:finite-key-affine-pmbff}, \(f_T\) lower-bounds the
one-round conditional von Neumann entropy for every physical realization
compatible with the complete one-round behavior and the source bound.
The source bound is imposed on the full emitted ensemble in each round,
and Proposition~\ref{prop:channel-lift-source-assumptions} lifts it to
any Bob--Eve extension produced by an input-independent channel. Thus the
one-round maps of the sequential protocol satisfy the entropy lower-bound
condition required of a min-tradeoff function. The selected generalized
entropy-accumulation theorem then gives
Eq.~\eqref{eq:finite-key-geat-form}
\cite{MetgerFawziSutterRenner2024GEAT,MetgerRenner2023GEATQKDSecurity}.
The generalized framework applies directly to prepare-and-measure QKD
protocols, so no entanglement-based reformulation is needed
\cite{MetgerRenner2023GEATQKDSecurity}. Finally, Eq.~\eqref{eq:finite-key-length}
is the standard leftover-hash step after subtracting the public
error-correction leakage. This proves that the accepted PM-BFF dual certificates provide the
finite-key input once the protocol-level finite-size choices are fixed.
\end{proof}

Equations~\eqref{eq:finite-key-geat-form} and
\eqref{eq:finite-key-length} are not evaluated numerically in this work
and should not be read as finite-\(n\) rate curves for the present
implementation.
They identify the object that must be retained from the SDP solver: the
feasible dual coefficients defining the affine function \(f_T\), not only
the reported entropy value. The Brown--Fawzi--Fawzi variational
construction is compatible with this use because the finite-node PM-BFF
certificate is itself a lower bound on the conditional von Neumann
entropy
\cite{BrownFawziFawzi2021ConditionalEntropies,BrownFawziFawzi2024ConditionalVonNeumannEntropy}.
Rényi-entropy accumulation refinements and finite-size
prepare-and-measure analyses provide possible sharper instantiations,
but no finite-size optimization is performed here
\cite{ArqandHahnTan2025RenyiEAT,KaminArqandGeorgeLuetkenhausTan2025FiniteSizePM}.

\section{Operational origins of the source assumptions}
\label{app:operational-origins}

This appendix records three physical routes to the operational source assumptions: optical photon-number constraints, bounded Hilbert-space dimension, and Bell remote preparation.

\subsection{Photon-number constraints and phase-randomized weak coherent sources}
\label{app:wcs-origin}

We assume trusted complete randomization of the global optical phase, with the random phase and every corresponding record unavailable to the receiver and Eve. A bound on the mean value of a trusted observable provides a natural alternative to a Hilbert-space dimension assumption. In optical prepare-and-measure protocols, the relevant carrier space is the infinite-dimensional Fock space and the mean photon number is a directly motivated choice \cite{VanHimbeeckWoodheadCerfGarciaPatronPironio2017NaturalAssumptions}. Bounds on the vacuum or non-vacuum contribution give closely related source restrictions and can be converted into bounds on the one-shot information available from the preparation ensemble \cite{PauwelsPironioTavakoli2025InformationCapacity}. We now show that, after trusted randomization of the global optical phase, the same physical restriction also bounds state exclusion and hence the normalized composites used in the main text.

Consider four optical preparations with total-photon-number distributions $p_n(x)$. Complete randomization of the global optical phase removes coherence between distinct total-photon-number sectors. The emitted state can therefore be written as
\begin{equation}
\label{eq:app-phase-randomized-source}
\rho_x
=
\sum_{n=0}^{\infty}p_n(x)\rho_x^{(n)},
\qquad
\rho_x^{(0)}
=
\ket{\mathrm{vac}}\!\bra{\mathrm{vac}},
\end{equation}
where the vacuum state is independent of the value $x$. For a weak coherent source of intensity $\mu_x$,
\begin{equation}
\label{eq:app-wcs-poisson}
p_n(x)
=
e^{-\mu_x}\frac{\mu_x^n}{n!},
\end{equation}
so continuous randomization of the global optical phase makes the coherent pulse an incoherent Poisson mixture of Fock sectors \cite{CaoZhangLoMa2015DiscretePhaseWCS}.

The relevant consequence is a common $x$-independent component. Let
\begin{equation}
\label{eq:app-common-component}
\rho_x
=
q\sigma+(1-q)\widetilde\rho_x,
\end{equation}
where $\sigma$ is a normalized state independent of $x$. For four-state discrimination, the common component contributes a measurement-independent probability $1/4$:
\begin{align}
D_{\mathcal Q}(\mathcal E)
&=
\frac14
\max_{\{N_x\}}
\sum_x
\Tr\!\left[\bigl(q\sigma+(1-q)\widetilde\rho_x\bigr)N_x\right]
\nonumber\\
&=
\frac q4
+
(1-q)D_{\mathcal Q}(\widetilde{\mathcal E}).
\label{eq:app-common-D}
\end{align}
For parity discrimination, each parity class contains two values of $x$, and hence
\begin{equation}
\label{eq:app-common-Pi}
\Pi_{\mathcal Q}(\mathcal E)
=
\frac q2
+
(1-q)\Pi_{\mathcal Q}(\widetilde{\mathcal E}).
\end{equation}
For state exclusion, each reported value is correct for three of the four preparations. Therefore,
\begin{align}
A_{\mathcal Q}(\mathcal E)
&=
\frac14
\max_{\{G_z\}}
\sum_x\sum_{z\neq x}
\Tr\!\left[\bigl(q\sigma+(1-q)\widetilde\rho_x\bigr)G_z\right]
\nonumber\\
&=
\frac{3q}{4}
+
(1-q)A_{\mathcal Q}(\widetilde{\mathcal E}).
\label{eq:app-common-A}
\end{align}
Combining Eqs.~\eqref{eq:app-common-D} and \eqref{eq:app-common-A} gives the exact decomposition
\begin{align}
(D\oplus A)_{\mathcal Q}(\mathcal E)
&=
\frac q2
+
(1-q)(D\oplus A)_{\mathcal Q}(\widetilde{\mathcal E})
\nonumber\\
&\leq
1-\frac q2.
\label{eq:app-common-DA}
\end{align}
Likewise,
\begin{align}
(\Pi\oplus A)_{\mathcal Q}(\mathcal E)
&=
\frac{5q}{8}
+
(1-q)(\Pi\oplus A)_{\mathcal Q}(\widetilde{\mathcal E})
\nonumber\\
&\leq
1-\frac{3q}{8}.
\label{eq:app-common-PiA}
\end{align}
These are direct bounds on the normalized composites. Both bounds are tight given only the common-component weight $q$: the excess sector may encode the four values of $x$ into mutually orthogonal states, for which $D=\Pi=A=1$. The operational source assumption constrains only the corresponding normalized sum and does not impose the component bounds separately.

For a phase-randomized source whose non-vacuum contribution satisfies
\begin{equation}
\label{eq:app-nonvacuum-bound}
\Tr\!\left[
\bigl(I-\ket{\mathrm{vac}}\!\bra{\mathrm{vac}}\bigr)\rho_x
\right]
\leq
\omega
\qquad
\forall x,
\end{equation}
the common vacuum weight obeys $q\geq1-\omega$. Equation~\eqref{eq:app-common-DA} then yields
\begin{equation}
\label{eq:app-nonvacuum-DA}
(D\oplus A)_{\mathcal Q}(\mathcal E)
\leq
\frac{1+\omega}{2}.
\end{equation}
A mean-photon-number bound gives a slightly coarser version of the same conclusion. For a phase-randomized source,
\begin{equation}
1-p_0(x)
=
\sum_{n\geq1}p_n(x)
\leq
\sum_{n\geq1}n p_n(x)
=
\Tr(\hat n\rho_x).
\end{equation}
Thus the common vacuum weight obeys $q\geq\max\{0,1-\bar n_{\max}\}$. Consequently,
\begin{equation}
\label{eq:app-mean-photon-DA}
\begin{aligned}
(D\oplus A)_{\mathcal Q}(\mathcal E)
&\leq
1-\frac12\max\{0,1-\bar n_{\max}\},
\\
(\Pi\oplus A)_{\mathcal Q}(\mathcal E)
&\leq
1-\frac38\max\{0,1-\bar n_{\max}\}.
\end{aligned}
\end{equation}
For $\bar n_{\max}\leq1$, these reduce to $(D\oplus A)_{\mathcal Q}\leq(1+\bar n_{\max})/2$ and $(\Pi\oplus A)_{\mathcal Q}\leq5/8+3\bar n_{\max}/8$. This is the direct connection between the photon-number assumption and the exclusion-assisted source-task values.

For phase-randomized weak coherent preparations whose total emitted mean photon numbers, including every emitted degree of freedom capable of carrying the value of $x$, satisfy $\mu_x\leq\mu_{\max}$, Eq.~\eqref{eq:app-wcs-poisson} gives
\begin{equation}
p_0(x)=e^{-\mu_x}\geq e^{-\mu_{\max}}.
\end{equation}
Extracting the common vacuum weight $q=e^{-\mu_{\max}}$ in Eq.~\eqref{eq:app-common-component} gives
\begin{equation}
\label{eq:app-wcs-DA}
(D\oplus A)_{\mathcal Q}(\mathcal E)
\leq
1-\frac12e^{-\mu_{\max}}.
\end{equation}
The same argument bounds all optimized source task values considered in this work:
\begin{equation}
\label{eq:app-wcs-bundle}
\begin{aligned}
D_{\mathcal Q}(\mathcal E)
&\leq
1-\frac34e^{-\mu_{\max}},
\\
\Pi_{\mathcal Q}(\mathcal E)
&\leq
1-\frac12e^{-\mu_{\max}},
\\
A_{\mathcal Q}(\mathcal E)
&\leq
1-\frac14e^{-\mu_{\max}},
\\
(D\oplus A)_{\mathcal Q}(\mathcal E)
&\leq
1-\frac12e^{-\mu_{\max}},
\\
(\Pi\oplus A)_{\mathcal Q}(\mathcal E)
&\leq
1-\frac38e^{-\mu_{\max}}.
\end{aligned}
\end{equation}
Given trusted global-phase randomization, the calibrated common vacuum contribution bounds every source-task value used in the manuscript.

When the four intensities are equal, $\mu_x=\mu$, the photon-number weights are independent of the value of $x$. Since distinct photon-number sectors are orthogonal, the receiver may first resolve the total photon number and then use an optimal decoder in that sector. For every $T\in\{D,\Pi,A,D\oplus A,\Pi\oplus A\}$, consequently,
\begin{equation}
\label{eq:app-wcs-sector-decomposition}
T_{\mathcal Q}(\mathcal E_{\mu})
=
\sum_{n=0}^{\infty}
e^{-\mu}\frac{\mu^n}{n!}
T_{\mathcal Q}(\mathcal E_n).
\end{equation}
Equation~\eqref{eq:app-wcs-bundle} is the universal consequence obtained by retaining only the value of $x$-independent vacuum sector; additional information about the non-vacuum encodings can tighten these bounds.

The derivation requires the phase-randomized source description above or another trusted mechanism that removes coherence between distinct total-photon-number sectors. Known-phase or retained-phase-record implementations require a different source analysis \cite{LoPreskill2007NonrandomPhases}. A mean-photon-number expectation alone does not imply the common-component decomposition in Eq.~\eqref{eq:app-common-component}. The bounds concern Alice's emitted ensemble before any untrusted loss-conditioned postselection and must cover every emitted degree of freedom capable of carrying the value of $x$. These physical assumptions are used only to derive a scalar operational ceiling; the classical-frontier and security analyses impose that ceiling itself and do not otherwise assume a weak coherent source model.

\subsection{Dimension implies a guessing bound}

For four equiprobable states $\rho_x$ supported on a $d$-dimensional Hilbert space and any four-state discrimination POVM $\{N_x\}$,
\begin{equation}
\frac14\sum_x\Tr(\rho_xN_x)
\le
\frac14\sum_x\Tr N_x
=
\frac d4,
\end{equation}
because $0\preceq\rho_x\preceq I_d$ and $\sum_xN_x=I_d$. Optimizing over the POVM gives $D\le d/4$. This implication is one-way: the operational source assumption does not require a dimension-$d$ realization.

\subsection{Bell remote preparation implies parity obliviousness}

Consider a tripartite state $\rho^{ABE}$ and binary measurements $\{A_{a|s}\}_{a=0}^{1}$ on Alice's system. Her event $a|s$ remotely prepares the subnormalized Bob--Eve state
\begin{equation}
\sigma_{a|s}^{BE}
=
\Tr_A[(A_{a|s}\otimes I_{BE})\rho^{ABE}].
\end{equation}
No-signalling is the operator identity
\begin{equation}
\sum_a\sigma_{a|0}^{BE}
=
\sum_a\sigma_{a|1}^{BE}.
\end{equation}
For unbiased outcomes define $\varrho_{a|s}^{BE}=2\sigma_{a|s}^{BE}$ and relabel $x=(a,a\oplus s)$. The no-signalling identity becomes
\begin{equation}
\varrho_{00}^{BE}+\varrho_{11}^{BE}
=
\varrho_{01}^{BE}+\varrho_{10}^{BE},
\end{equation}
which is the exact joint Bob--Eve parity relation used at $\Pi=1/2$.

Let Bob use binary measurements $B_{b|y}$. The conditioned prepare-and-measure behavior is $q(b|x(a,s),y)=2p(a,b|s,y)$. Since $x_y=a\oplus sy$,
\begin{equation}
\begin{aligned}
R(q)
&=\frac18\sum_{a,s,y}q(b=a\oplus sy|x(a,s),y)\\
&=\frac14\sum_{s,y}p(a\oplus b=sy|s,y),
\end{aligned}
\end{equation}
so the RAC score equals the CHSH winning probability. For $s=0$, the states $00$ and $11$ correspond to Alice's outcomes $a=0,1$ and form the matched key branch. This establishes the analytical DI-to-SDI connection without making Bell structure an assumption of the operational protocol \cite{WrightFarkas2023BellContextualityMap}.

\bibliography{cite}

\end{document}